\documentclass[letterpaper,twocolumn,10pt]{article}
\usepackage{usenix-2020-09}

\usepackage{amsmath}
\usepackage{amsfonts}
\usepackage{amssymb}
\usepackage{amsthm}
\usepackage{caption}
\usepackage{color}
\usepackage{comment}
\usepackage{enumitem}
\usepackage{extarrows}
\usepackage{float}
\usepackage{framed}
\usepackage{graphicx}
\usepackage{hyperref}
\usepackage{listings}
\usepackage{mathtools}
\usepackage{nicefrac}
\usepackage{soul}
\usepackage{subcaption}
\usepackage{xurl}

\definecolor{sqlbgcolor}{cmyk}{0, 0, 0.4, 0}

\newcommand{\sys}{Membrane}
\newcommand{\sqlbg}[1]{{\setlength{\fboxsep}{1pt} \colorbox{sqlbgcolor}{$#1$}}}
\newcommand{\sqlview}{AC view}
\newcommand{\asqlview}{an AC view}

\newcommand{\sqlviews}{AC views}
\newcommand{\sqlViews}{AC Views}

\newcommand{\viewfamily}{AC view family}
\newcommand{\aviewfamily}{an AC view family}

\newcommand{\viewfamilies}{AC view families}

\newcommand{\Viewfamilies}{AC view families}

\newcommand{\secref}[1]{\S{}\ref{#1}}
\newcommand{\appref}[1]{Appendix~\ref{#1}}
\newcommand{\figref}[1]{Fig.~\ref{#1}}
\newcommand{\tabref}[1]{Table~\ref{#1}}

\newcommand{\queryref}[1]{Query~\ref{#1}}

\newcommand{\parhead}[1]{\noindent\textbf{#1.}}

\newcommand{\rwe}{\textsf{RWE}}
\newcommand{\yellow}{\textsf{yellow}}
\newcommand{\storesales}{\textsf{store\_sales}}

\newcommand{\rwestate}{\textbf{rwe\_state}}
\newcommand{\medication}{\textbf{medication}}
\newcommand{\diagnosis}{\textbf{diagnosis}}

\newcommand{\dropoffpickup}{\textbf{dropoff\_pickup}}
\newcommand{\salesstore}{\textbf{sales\_store}}
\newcommand{\salesdate}{\textbf{sales\_date}}
\newcommand{\rweobsstate}{\textbf{rwe\_obs\_state}}

\newcommand{\rweineqobs}{\textbf{rwe\_ineq\_obs}}
\newcommand{\rweineqstate}{\textbf{rwe\_ineq\_state}}
\newcommand{\rweineqor}{\textbf{rwe\_ineq\_or}}
\newcommand{\rweineqand}{\textbf{rwe\_ineq\_and}}

\newtheorem{theorem}{Theorem}
\newtheorem{definition}{Definition}

\newcommand{\prf}[2]{\mathsf{PRF}(#1,\,#2)}
\newcommand{\enc}[2]{\mathsf{Enc}(#1,\,#2)}
\newcommand{\encsymb}[0]{\mathsf{Enc}}
\newcommand{\dec}[2]{\mathsf{Dec}(#1,\,#2)}
\newcommand{\onetimeenc}[2]{\mathsf{OTE}(#1,\,#2)}
\newcommand{\onetimeencsymb}[0]{\mathsf{OTE}}

\newcommand{\concatnospace}[0]{||}
\newcommand{\concat}[0]{\,\concatnospace\,}
\newcommand{\definedas}[0]{\coloneqq}

\newcommand{\projkey}[1]{k^{\mathsf{proj}}_{#1}}
\newcommand{\selkey}[1]{k^{\mathsf{sel}}_{#1}}
\newcommand{\selkeyset}[1]{K^{\mathsf{sel}}_{#1}}

\newcommand{\view}[0]{\mathsf{view}}
\newcommand{\viewkeycomp}[1]{k^{\view{}}_{#1}}
\newcommand{\viewkey}[0]{k^{\view{}}}
\newcommand{\viewkeyset}[0]{K^{\mathsf{view}}}
\newcommand{\family}[0]{\mathsf{fam}}
\newcommand{\famkey}[0]{k^{\family{}}}
\newcommand{\tabkey}[0]{k^{\mathsf{tab}}}
\newcommand{\counter}[1]{\mathsf{count}_{#1}}

\newcommand{\predkey}[1]{k^\mathsf{pred}_{#1}}

\newcommand{\decrightarrow}[1]{\mathrel{\raisebox{-3pt}{$\xrightarrow{#1}$}}}

\newcommand{\rowcolor}[1]{{\color{teal} #1}}
\newcommand{\colcolor}[1]{{\color{olive} #1}}
\newcommand{\predcolor}[1]{{\color{red} #1}}

\newcommand{\sqlwildcard}[0]{\mathsf{*}}
\newcommand{\sqlselect}[0]{{\color{blue}\texttt{SELECT }}}
\newcommand{\sqlnsselect}[0]{{\color{blue}\texttt{SELECT}}}
\newcommand{\sqlfrom}[0]{{\color{blue}\texttt{ FROM }}}
\newcommand{\sqlwhere}[0]{{\color{blue}\texttt{ WHERE }}}

\newcommand{\sqlrswhere}[0]{{\color{blue}\texttt{WHERE }}}
\newcommand{\sqlnswhere}[0]{{\color{blue}\texttt{WHERE}}}
\newcommand{\sqlin}[0]{{\color{blue}\texttt{ IN }}}
\newcommand{\sqlnsin}[0]{{\color{blue}\texttt{IN}}}
\newcommand{\sqlnotin}[0]{{\color{blue}\texttt{ NOT IN }}}
\newcommand{\sqlnsnot}[0]{{\color{blue}\texttt{NOT}}}

\newcommand{\sqland}[0]{{\color{blue}\texttt{ AND }}}
\newcommand{\sqlnsand}[0]{{\color{blue}\texttt{AND}}}
\newcommand{\sqlor}[0]{{\color{blue}\texttt{ OR }}}
\newcommand{\sqlnsor}[0]{{\color{blue}\texttt{OR}}}
\newcommand{\sqllower}[1]{{\color{blue}\texttt{LOWER}}\texttt{(}#1\texttt{)}}
\newcommand{\sqlend}[0]{{\color{blue}{\texttt{;}}}}

\newcommand{\encfs}{blaze1993cfs,kallahalla2003plutus,cloudproof,goh2003sirius,keybase, tresorit,spideroak,sync,icloud_security}
\newcommand{\encdbsys}{hacigumus2002executing, popa2011cryptdb,  pappas2014blindseer, papadimitriou2016seabed, shafagh2015talos, shafagh2017pilatus, poddar2019arx, cash2014dynamic, kamara2018sql, burkhalter2020timecrypt,dauterman2022waldo,bater2017smcql,kamara2020optimal,cash2021improved}

\newcommand{\encsearchsys}{song2000searchable,jarecki2013ospir,cash2014dynamic,curtmola2006improved,dauterman2020dory,stefanov2014dsse,jutla2022efficient,chase2010structured,pappas2014blindseer}

\newcommand{\cryptacsys}{wang2016sieve, kumar2019jedi, burkhalter2020timecrypt, hang2015enki, sarfraz2015dbmask}

\newcommand{\adv}{\mathcal{A}}
\newcommand{\advsel}{\mathcal{A}_{\mathsf{sel}}}
\newcommand{\advstatic}{\mathcal{A}_{\mathsf{static}}}
\newcommand{\chl}{\mathcal{C}}
\newcommand{\hyb}[1]{\mathcal{H}_{#1}}
\newcommand{\secp}{\lambda{}}
\newcommand{\reltup}[1]{$\mathcal{R}_{#1}$}

\DeclareMathAlphabet{\mathcal}{OMS}{cmsy}{m}{n}

\usepackage{listings}
\usepackage[frozencache]{minted}
\definecolor{codegreen}{rgb}{0,0.6,0}
\usepackage[explicit]{titlesec}
\titlespacing*{\section}{0pt}{4pt plus 1pt minus 0pt}{2pt plus 0pt minus 0pt}
\titlespacing*{\subsection}{0pt}{3pt plus 1pt minus 0pt}{1pt plus 0pt minus 0pt}
\titlespacing*{\subsubsection}{0pt}{2pt plus 1pt minus 0pt}{1pt}

\usepackage{titling}
\usepackage{caption}
\microtypesetup{spacing=false}

\begin{document}

\setlength{\abovedisplayskip}{1pt}
\setlength{\belowdisplayskip}{1pt}
\setlength{\abovedisplayshortskip}{0pt}
\setlength{\belowdisplayshortskip}{0pt}

\setlength{\jot}{0pt}

\date{}

\title{\Large \bf \sys{}: A Cryptographic Access Control System for Data Lakes}

\author{
{\rm Sam Kumar}\\
\textit{UCLA}
\and
{\rm Samyukta Yagati}\\
\textit{UC Berkeley}
\and
{\rm Conor Power}\\
\textit{UC Berkeley}
\and
{\rm David E. Culler}\\
\textit{Google}
\and
{\rm Raluca Ada Popa}\\
\textit{UC Berkeley}
}

\maketitle

\begin{abstract}
Organizations use data lakes to store and analyze sensitive data.
But hackers may compromise data lake storage to bypass access controls and access sensitive data.
To address this, we propose \sys{}, a system that (1) cryptographically enforces data-dependent access control \emph{views} over a data lake, (2) without restricting the analytical \emph{queries} data scientists can run.
We observe that data lakes, unlike DBMSes, disaggregate computation and storage into separate trust domains, making at-rest encryption sufficient to defend against remote attackers targeting data lake storage, even when running analytical queries in plaintext.
This leads to a new system design for \sys{} that combines encryption at rest with SQL-aware encryption.
Using block ciphers, a fast symmetric-key primitive with hardware acceleration in CPUs, we develop a new SQL-aware encryption protocol well-suited to at-rest encryption.
\sys{} adds overhead only at the \emph{start} of an interactive session due to decrypting views, delaying the first query result by up to $\approx 20\times$; subsequent queries process decrypted data in plaintext, resulting in low amortized overhead.

\end{abstract}

\section{Introduction}\label{sec:intro}

Data lakes have emerged as a central paradigm for data analysis.
Their key innovation, compared to DBMSes, is to \emph{separate} compute resources (e.g., EC2) from storage resources (e.g., S3)~\cite{melnik2010dremel}.
This has two benefits.
First, it allows compute and storage to be scaled independently.
Second, it allows data scientists to analyze data using their framework of choice (e.g., Pandas, Spark, etc.).
Companies like Microsoft, Google, and Databricks provide data lake platforms~\cite{azuredatalake, biglake, deltalake} based on open file formats (e.g., Parquet~\cite{parquet}) in cloud storage.

Data lakes are increasingly used for sensitive data, like financial data~\cite{financialsummittalks2022} or healthcare data~\cite{healthcaresummittalks2022}.
Thus, each data lake user (e.g., data scientist) must be granted access to only the data that she needs.
The state of the art is to use cell-level, \emph{data-dependent} access policies.
For example, the pharmaceutical company Eisai demonstrated granting each data scientist access to patient data in only certain US States~\cite{esaisummittalk2021}.
This access policy is described by a SQL query/view such as:
\begin{equation}
\hspace{-1.4ex}\sqlbg{\sqlselect{}\sqlwildcard{}\sqlfrom{}\mathsf{RWE}\sqlwhere{}\mathsf{State}\sqlin{} (\mathsf{``IA"}, \mathsf{``IL"})\sqlend{}}\hspace{-0.5ex}
\label{query:intro_example}
\end{equation}
We refer to access control (AC) policies based on SQL views as \emph{\sqlviews{}}.
For data lakes, such access control mechanisms are provided directly by data lake platforms~\cite{securitysummittalk2021} and by third-party companies like Immuta~\cite{immutadatabricks, immutasummittalk2021} and Privacera~\cite{privaceradatabricks}.

Alas, data theft from cloud storage is common~\cite{s3hacks2017, verizonbreach2017, airportbreach2022}.
Access control mechanisms help, but are not foolproof.
In 2019, for example, an attacker
bypassed CapitalOne's Web Application Firewall (WAF) and gained access to sensitive data in Amazon S3, including 140,000 Social Security numbers and 80,000 bank account numbers~\cite{capitalonebreach2019}.

In light of such attacks, it is desirable to encrypt the data lake.
This would keep the data protected even if an attacker breaks into the storage.
Na\"ively using encryption, however, is weak; since all data scientists have the secret key, a \emph{single} compromised data scientist would undermine encryption for the entire data lake.
Instead, we want to \emph{cryptographically enforce access control}~\cite{\cryptacsys}, so each data scientist's key can decrypt only data she is allowed to access.

\textbf{How can we design a cryptographic access control system for data lakes?}
We answer this question with our system \sys{}.
\sys{} encrypts data at the storage servers to remove storage from the Trusted Computing Base (TCB), and ensures that each data scientist can only decrypt and analyze data matching the \sqlviews{} she is granted.
To achieve this, \textit{\sys{} is the first system to combine encryption at rest with SQL-aware encryption}.
At-rest encryption gives data scientists \ul{flexibility} to run arbitrary analytics tools like Pandas or SQL, and SQL-aware encryption is used to cryptographically enforce fine-grained, data-dependent \ul{access control} views.

\smallskip
\parhead{New system model}
One may try using an encrypted file system (EFS)~\cite{\encfs}, encrypted search system (ESS)~\cite{\encsearchsys}, or encrypted database (EDB)~\cite{\encdbsys}.
However, data lakes have two requirements, \ul{flexibility} and \ul{access control}, that render such designs unsuitable.

First, data scientists expect the \ul{flexibility} to analyze data using \emph{arbitrary} SQL queries or data science/ML frameworks.
In contrast, practical EDBs/ESSes limit users to only the small subset of SQL supported cryptographically, and EFSes do not support data-dependent queries (SQL) at all.

\sys{} achieves flexibility via at-rest encryption.
Data scientists first download encrypted table(s) from storage to their compute nodes, then decrypt the portion of the encrypted table(s) matching \sqlviews{} that they are granted, and, finally, run analytics on decrypted data \emph{in cleartext} at their compute nodes.
Crucially, data analysis is done over \emph{unencrypted} data, so data scientists have the flexibility to issue \emph{arbitrary} SQL queries and use any analytics framework (Pandas, Spark, etc.).
Only \sqlviews{}, usually simpler than analytical queries, are constrained by cryptography.
We call this the \textbf{Encrypted Data Lake (EDL) model}.

In using at-rest encryption, our insight is that data lakes' \emph{separation between compute and storage} can enhance the value of at-rest encryption.
To understand how, let us contrast data lakes with DBMSes.
In a traditional DBMS, at-rest encryption protects against attackers who only steal a disk drive; in the context of a remote attacker who compromises software, compute and storage are \emph{in the same trust domain} (e.g., the same server).
Protecting against such remote attackers when processing data in the cloud requires also removing compute from the TCB, like an EDB; this requires data analytics to compute on encrypted data, limiting EDBs' flexibility to only queries supported via cryptography~\cite{papadimitriou2016seabed}.
But in a data lake, compute resources are physically and logically distinct from storage resources, as they are often developed/administered by different engineering teams~\cite{power2021cosmos}.
For example, a remote attacker who gains access to an organization's storage (e.g., S3 buckets) has not necessarily compromised its compute (e.g., EC2 VMs).
Thus, the separation of compute and storage enables \sys{}'s at-rest encryption to protect against remote attackers who compromise software.
Specifically, \sys{} (1) fully removes storage from the TCB, and (2) trusts a data scientist's compute with \emph{only} data she is granted access to.

In a sense, EFSes/EDBs/ESSes have stronger security than the EDL model, because they hide data from both storage and compute at the cloud provider.
So, it may seem natural to adapt EFSes/EDBs/ESSes  to the EDL model by weakening their security.
Concretely, one could run an EFS/EDB/ESS server in storage and the client in compute, and have data scientists ``query'' storage for data in their \sqlviews{}.
This fails due to data lakes' second requirement: \ul{access control}.
Many EDBs/ESSes do not support access control.
EFSes, and the few EDBs/ESSes with access control, only allow access policies based on \emph{public} attributes like file paths, not \emph{private} data like the \textsf{State} field (\queryref{query:intro_example}).
Further, this approach requires running a cryptographic protocol in storage, not supported by current cloud storage offerings (e.g., AWS S3).

\sys{} delivers access control via a \textbf{new SQL-aware encryption protocol} compatible with existing cloud storage offerings.
This is possible because \sys{} solves a different problem than EDBs.
While EDBs compute the \emph{encrypted} result of a SQL query, \sys{}'s protocol produces the \emph{plaintext} data of a SQL \sqlview{}, using a decryption key for the view.
\sys{}'s access pattern leakage is limited to which table partitions are fetched from storage.
As data lakes often use fast intra-datacenter networks, one can fetch all of a table's partitions to hide access patterns, if needed (\secref{s:linear_scan}).

\smallskip
\parhead{Designing the protocol}
There are several major challenges in designing \sys{}'s new SQL-aware protocol.

First, there can be many \sqlviews{}, so it is untenable to add per-row data for each one.
Our insight is that we can add a \emph{single} unit of per-row cryptographic material for \emph{a large set of \sqlviews{}}.
We refer to such sets of \sqlviews{} as \emph{\viewfamilies{}} and represent them as query templates.
For example, the \viewfamily{} $\sqlbg{\sqlselect{}\mathsf{*}\sqlfrom{}\mathsf{RWE}\sqlwhere{}\mathsf{State}\sqlin{}?x\sqlend{}}$ includes all $2^{50}$ \sqlviews{} for sets of States.
Many \sqlviews{} are often captured by a few \viewfamilies{}.

Second, we must map SQL to cryptography.
Existing EDBs compute a query plan and have a subprotocol for each step, but this reveals at which step (e.g., which predicate in the \sqlnswhere{} clause) each row is filtered out.
To avoid leaking partial results, \sys{} instead \emph{rewrites} SQL queries into a single monolithic operator called \emph{\sys{}-canonical form}.

Third, due to the large scale of data lakes, cryptographic processing must be fast, ideally gigabytes per second.
This restricts us to fast symmetric-key tools like the block cipher AES, which, on its own, cannot support SQL and is limited to equality checks.
Our insight is to apply an arbitrary function $g$ to cell data \emph{before} applying the block cipher, enabling clauses of the form $\sqlbg{g(\mathsf{row})\sqlin{}?x}$.
This may still seem limiting on first glance, but by carefully choosing $g$, we can actually rewrite \emph{inequalities} (e.g., $<$, $\neq$), \sqlnsand{}s, and \sqlnsor{}s into this form.

Fourth, because \sys{} supports access policies based on \emph{private, encrypted} data, the client cannot easily know which rows of a table they can decrypt.
To solve this, we develop \emph{key-hiding tags} that allow a user to identify rows to decrypt \emph{up to 50,000$\times$} faster than na\"ively trying each row.

\smallskip
A 52-core server can decrypt \asqlview{} over 200 GB of in-memory patient data in 2--15 seconds using \sys{}.
Key-hiding tags provide a speedup of up to 50,000$\times$ and are crucial for achieving ``big data'' speeds.
The full process of downloading and decrypting a view takes 30--100 seconds.
A limitation of \sys{} is that size overheads are up to an order of magnitude, due almost entirely to losing compression when encrypting data.
\sys{} decrypts views at the \emph{start} of an interactive session; in a PySpark-SQL setup in Databricks, it increases the time to completion for the first analytical query by $\approx 20\times$ compared to non-cryptographic \sqlviews{}.
\sys{} is designed so that subsequent analytical queries in the session can use the already-decrypted view, with \emph{no} overhead from \sys{}.
Thus, its amortized overhead for an interactive session can be small.

\section{System Overview}
\label{sec:model}
\label{s:datalakes}
\label{s:lifecycle}
\label{s:cleaning}

\begin{figure}[t]
  \centering
  \includegraphics[width=\linewidth, trim=0.8cm 0cm 0.7cm 0cm, clip]{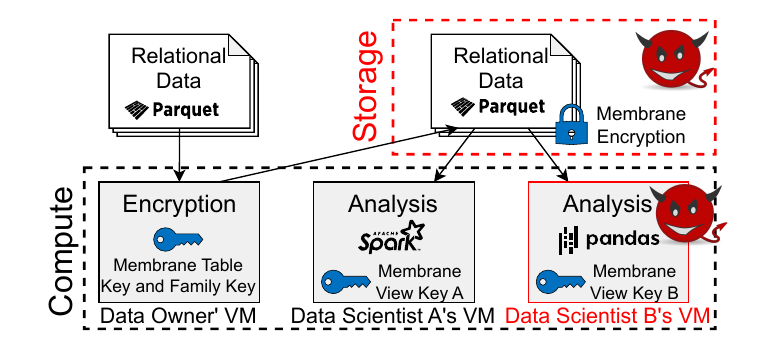}
  \caption{A data lake and how \sys{} integrates with it. \sys{} components are in blue, and our threat model is in red.}
  \label{fig:data_lake}
\end{figure}

We consider a system model where a \emph{data owner} has a sensitive dataset and wishes to allow \emph{data scientists} to run analytics jobs against the dataset.
The data owner specifies what part of the dataset each data scientist can access as an \emph{AC view}.
\sys{} enables the data owner to (1) encrypt the dataset, and (2) grant each data scientist cryptographic \emph{view key(s)} that can decrypt only data described by her AC view(s).
The data owner can later add rows without re-generating view keys.

The data owner places the encrypted dataset in a data repository, and data scientists run analytics jobs against the data by allocating compute resources distinct from storage.
For example, if the data repository is in cloud storage, data scientists may spawn VMs in the same cloud region for analytics jobs.
Data scientists analyze data using an \emph{analytics framework} (e.g., Spark-SQL, Pandas, etc.).
To process a compute job (e.g., SQL query for Spark-SQL), the framework downloads the relevant data from storage to the compute servers and processes the data there.
It can keep data in memory for an interactive session to avoid rereading it on each job.

The overall data lake consists of multiple \sys{} deployments, each with a different owner.
Each deployment contains some data, and the data owner is a privileged employee (or team) that determines who can access what for that deployment.
The data scientists to whom access is granted may be in different teams, or even different organizations altogether.

\subsection{Applying \sys{} to Data Lakes}
\label{sec:background}

Data lakes are varied in applications and deployment models.
For example, data lakes are used with unstructured data, such as raw logs, or multimodal data including images and video.

\sys{} requires the dataset to be in a structured, relational form.
Any unstructured data must first be converted to a structured form (i.e., \emph{cleaned}) before the data owner can apply \sys{}.
This requirement is fundamental to access control---unstructured data are not in a consistent format, so it is difficult to programmatically enforce access control, cryptographically or otherwise.
Moreover, this requirement is consistent with many data lake use cases.
Data lake platforms~\cite{deltalake, azure_synapse_parquet, power2021cosmos} are built on file formats for structured, relational data like Parquet~\cite{parquet} and ORC, and data lake access control offerings~\cite{immutasummittalk2021, privaceradatabricks, esaisummittalk2021, databricksAccessControl} target structured (e.g., tabular) data.
The established \emph{medallion architecture}~\cite{deltapipeline} involves cleaning data ahead of time, as \sys{} requires.

Because \sys{} only protects structured data, the process of converting unstructured data to structured data, if applicable to a deployment, must be protected via other means.
For example, one can place the unstructured data in a staging area separate from the data lake (e.g., an on-premises cluster).
Alternatively, one can use the data lake for this process, but encrypt any unstructured data with a symmetric key held only by the data owner or the party carrying out the conversion.

\parhead{Strawman \#1}
To better motivate \sys{}'s system model, we compare it to a strawman based on a trusted AC service.
The strawman is to protect data in storage with symmetric-key encryption, placing the key at the AC service.
The AC service decrypts data and computes AC views for data scientists.
It is trusted to see data contents and enforce access control.

Such an AC service is a central point of attack, which \sys{} eliminates.
This is because an AC service must accept requests from untrusted parties and be highly available, making it much harder to firewall and harden than a data owner, who does not need to host an online service with the secret key.
This is the same reason that an EFS design, in which clients store their keys~\cite{goh2003sirius}, is preferable to having all clients access the file system via a single trusted proxy holding the secret keys.
Similar arguments motivate delegable access control~\cite{kumar2019jedi} and HTTPS, in which Certificate Authorities access root keys rarely to lessen the risk of compromise.

To preserve the benefits of cloud storage, an AC service would need to be engineered to have very high availability and scalability to match cloud storage offerings.
The alternative, to integrate the AC service into the storage endpoint, would make it a central point of attack \emph{for the entire data lake}.

\subsection{Threat Model and Security Guarantees}
\label{s:informal_security}

Against a malicious adversary $\adv{}$ who has compromised the storage servers and some data scientists (red, \figref{fig:data_lake}), \sys{} guarantees that $\adv{}$ cannot see data, except what compromised data scientists are permitted to see by their \sqlviews{}.
We provide a \ul{\em formal cryptographic security definition and proof} in \appref{app:formalism}.

In a data lake, large tables (e.g, those gigabytes in size or larger) are usually \emph{partitioned} into multiple Parquet/ORC files, each storing a subset of the table's rows.
If a data scientist fetches only some of a table's partitions, then $\adv{}$ learns which partitions were fetched.
As discussed in \secref{s:linear_scan}, one can hide row-level access patterns from $\adv{}$ by having data scientists fetch \emph{all} rows in a table when calling \textsf{RevealView}.
Even that does not hide \emph{which tables} they access, or \emph{when} they access those tables.
\sys{} does not hide a table's schema, the number of rows in a table or partition, or size of each cell.
\sys{} does not provide anonymity.
\sys{} does not hide the \emph{positions} of cells accessible to compromised data scientists.
To limit leakage via cell positions, one can shuffle rows in a table before encrypting with \sys{}.

\sys{} is designed to be used with existing techniques for strong integrity guarantees~\cite{li2004sundr, mahajan2010depot, hu2020ghostor}.
For example, the data owner may sign updates to files to prevent $\adv{}$ from changing them arbitrarily and sign the entire data repository using a Merkle tree to prevent $\adv{}$ from selectively rolling back files.
Such techniques are orthogonal to \sys{}'s core contributions and are easy to integrate with \sys{}.

\subsection{AC View Families}
\label{s:view_families}

The decryption keys given to data scientists in \figref{fig:data_lake} grant access to \sqlviews{} specified as SQL.
How can we craft such decryption keys?
While functional encryption~\cite{boneh2011functional} enables this in theory, it is impractically slow for general functions.

To circumvent this, \sys{} slightly relaxes the model: Before \sys{} can generate a decryption key for \asqlview{}, the encrypted table must first be augmented with some cryptographic material.
However, this must be designed carefully; simple approaches requiring adding per-AC-view data to each row are undesirable because there can be many views.

Our insight is that we can add a single unit of per-row cryptographic material for \emph{a large set of \sqlviews{}}.
We refer to such sets of \sqlviews{} as \emph{\viewfamilies{}} and refer to the process of adding this material as \emph{instantiating} a view family.
We represent \viewfamilies{} as SQL queries with constants replaced by wildcards.
For example, the \viewfamily{} $\sqlbg{\sqlselect{}\sqlwildcard{}\sqlfrom{}\mathsf{customers}\sqlwhere{}\mathsf{Location}\sqlin{}?x\sqlend{}}$ includes \sqlviews{} where $?x$ is replaced by any set of strings (e.g., $\sqlbg{\sqlrswhere{}\mathsf{Location}\sqlin{}(\mathsf{``Phoenix"},\mathsf{``Mesa"})}$).

Once \aviewfamily{} is instantiated for a table, it is possible to generate a decryption key for any \sqlview{} in that \viewfamily{}.
This is better than adding per-view state to each row because a single \viewfamily{} can describe many  \sqlviews{} (e.g., many values for $?x$).
This idea, to group views into a small number of patterns (i.e., view families), is inspired by non-cryptographic attribute-based access control~\cite{immutasummittalk2021}.

\parhead{Strawman \#2}
To better motivate view families, we consider a strawman design that materializes each \sqlview{} as a separate Parquet file, and then uses file-level encryption.
This has two drawbacks.
(1) It requires maintaining \emph{many copies of data} in the data lake.
Data lake providers generally avoid this because of the risks of some copies becoming stale and the costs of keeping multiple copies of data up to date~\cite{fansummittalk2022}.
(2) Materializing \sqlviews{} can have a large storage footprint, particularly if there is a large overlap among \sqlviews{}.

Instantiating \viewfamilies{} in \sys{} also requires space.
Unlike the strawman, \sys{}'s extra space scales in the number of \emph{view families}, not \emph{views}.
This is a large reduction, as one view family describes many possible views.

\subsection{\sys{}'s API and Workflow}\label{s:workflow}

\begin{figure}[t]
    \setlength{\FrameSep}{2pt}
    \begin{framed}
        \small
        \begin{itemize}[leftmargin=*,topsep=0ex]
            \item $\mathsf{EncryptTable}(t) \rightarrow t', \tabkey{}$
            \begin{itemize}[leftmargin=*,topsep=0ex,noitemsep]
                \item Input $t$: table to encrypt
                \item Outputs $t'$, $\tabkey{}$: encrypted table $t'$ and its \emph{table key} $\tabkey{}$
            \end{itemize}
            \item $\mathsf{AddFamily}(t, \tabkey{}, \family{}) \rightarrow t', \famkey{}$
            \begin{itemize}[leftmargin=*,topsep=0ex,noitemsep]
                \item Inputs $t$, $\tabkey{}$: an encrypted table $t$ and its \emph{table key} $\tabkey{}$
                \item Input $\family{}$: SQL describing a view family
                \item Output $t'$: new version of $t$, with $\mathsf{fam}$ instantiated
                \item Output $\famkey{}$: \emph{family key} for $t'$'s instantiation of $\family{}$
            \end{itemize}
            \item $\mathsf{ViewGen}(\view{}, \famkey{}) \rightarrow \viewkey{}$
            \begin{itemize}[leftmargin=*,topsep=0ex,noitemsep]
                \item Input $\view{}$: SQL describing a view in $\family{}$
                \item Input $\famkey{}$: \emph{family key} for an instantiation of $\family{}$
                \item Output $\viewkey{}$: \emph{view key} corresponding to $\view{}$
            \end{itemize}
            \item $\mathsf{RevealView}(t, \viewkey{}, \mathsf{fil}) \rightarrow t'$
            \begin{itemize}[leftmargin=*,topsep=0ex,noitemsep]
                \item Inputs $t$, \textsf{fil}: encrypted table $t$ and partition filter \textsf{fil}
                \item Input $\viewkey{}$: \emph{view key} from a family instantiated in $t$
                \item Output: $\view{}$ applied to $t$'s unencrypted data
            \end{itemize}
        \end{itemize}
    \end{framed}
    \vspace*{-1ex}
    \caption{Summary of \sys{}'s API.}
    \label{fig:api}
\end{figure}

\figref{fig:api} summarizes \sys{}'s API and describes the EDL model, which we formalize in \appref{app:formalism}.
Upon obtaining a table $t$ (e.g., a set of Parquet files), the data owner (1) \textsf{Encrypt}s $t$ and puts the result in storage; (2) uses \textsf{AddFamily} to instantiate their desired \viewfamilies{} in the encrypted table; and (3) calls \textsf{ViewGen} to generate \emph{view keys} for desired \sqlviews{}, and gives view keys to data scientists to grant them access.
The data owner can call \textsf{AddFamily} and \textsf{ViewGen} at any time to instantiate additional \viewfamilies{} and grant access to additional \sqlviews{}.
To run analysis jobs, data scientists call \textsf{RevealView} to decrypt \sqlviews{} they are granted.

We envision that data scientists will call \textsf{RevealView} \emph{once per table at the start of an interactive session} to materialize the views locally at the compute servers.
Then, they can run compute jobs on these materialized views in plaintext using whatever framework they wish.
Essentially, they incur the overhead of \textsf{RevealView} once and can then analyze the decrypted view in plaintext indefinitely with no overhead.
That said, it may be unavoidable to call \textsf{RevealView} again in some cases, so we still designed \textsf{RevealView} to be performant.

\sys{} does not support modifying data in place; modern data lakes (e.g., Lakehouse systems) often implement logical modifications as physical appends~\cite{deltalake, armbrust2021lakehouse}. 

Access revocation can be achieved, in principle, via \emph{lazy revocation}~\cite{kallahalla2003plutus}.
The principle is that one cannot make a data scientist ``forget'' data she was previously granted, but one can hide future rows added after revoking access.
Specifically, one can use new family keys when encrypting future rows, and generate fresh view keys from those family keys to grant to users who were not revoked.
With this design, old data remain visible to a revoked party, but new data are not.

\subsection{System Architecture}
\label{s:architecture}

\begin{figure}[t]
    \centering
    \includegraphics[width=\linewidth]{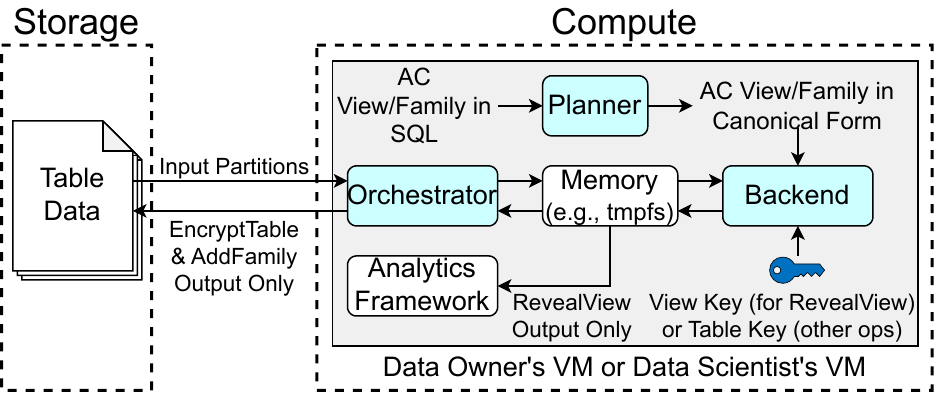}
    \caption{\sys{}'s architecture; its components are in blue.}
    \label{fig:architecture}
\end{figure}

\sys{} has a \emph{planner}, \emph{backend}, and \emph{orchestrator} (\figref{fig:architecture}).

\subsubsection{Planner (\secref{sec:planner})}
\label{s:motivation_canonical}

A critical design decision in \sys{} is to use a \emph{single} cryptographic protocol that supports only views/families of a particular form, which we call \emph{\sys{}-canonical form}.
\sys{} has a \ul{planner} that rewrites SQL views/families into \sys{}-canonical form, off the critical path.
This departs from prior systems, which directly support SQL by composing \emph{multiple} subprotocols for different subexpressions within a SQL query.

To understand why we design \sys{} this way, consider EDBs.
EDBs like CryptDB have cryptographic subprotocols for each operator (e.g., subexpressions of the \sqlnswhere{} clause) and compose them to execute a query.
Consider \queryref{query:edb_example}:
\begin{equation}
\hspace*{-0.25ex}
\sqlbg{\sqlselect{}\sqlwildcard{}\sqlfrom{}t\sqlwhere{}a\texttt{ = "foo"}\sqland{}b\texttt{ < 150};}
\hspace*{-0.25ex}
\label{query:edb_example}
\end{equation}
An EDB like CryptDB has separate subprotocols to check if a row matches $\sqlbg{a\texttt{ = "foo"}}$ (e.g., deterministic encryption) and if a row matches $\sqlbg{b\texttt{ < 150}}$ (e.g., order-preserving encryption), and would run both at the server to filter out rows.
Unfortunately, this does not work for \sys{} because a client would learn \emph{which predicates match even for rows outside of the \sqlview{}}, leaking information about $a$ and $b$ for those rows.

\sys{} avoids this issue because \sys{}-canonical form is supported in cryptography as a \emph{single monolith} that does not leak the results of subexpressions.

Choosing \sys{}-canonical form is tricky because it must be cryptographically efficient, yet general enough that complex SQL forms can be rewritten to it.
We identify the appropriate canonical form to have a \sqlnswhere{} clause as a disjunction of predicates $\sqlbg{g(\mathsf{row})\sqlin{}?x}$, where $g$ is an \emph{arbitrary function}.
This lets \sys{} support \textbf{$\sqlnsand{}$s and inequalities} and can be implemented with only fast, symmetric-key cryptography.

\subsubsection{Backend (\secref{sec:backend})}
\label{s:protocol_intuition}

The \ul{backend} executes \sys{}'s cryptographic protocol on in-memory data.
In \textsf{EncryptTable}/\textsf{AddFamily}/\textsf{RevealView}, \sys{}'s backend can process partitions in parallel on multiple CPUs, producing one output partition for each input partition.

To convey the essence of \sys{}'s protocol, we present the following \emph{highly simplified example}.
Take the view family $\sqlbg{\sqlselect{}\mathsf{bname}}\allowbreak\sqlbg{\sqlrswhere{}\allowbreak\mathsf{color}}\allowbreak\sqlbg{\sqlin{}?x\sqlend{}}$ for the table in \figref{fig:example_relation}.
In \textsf{AddFamily}, (1) for each color (blue, red, green) we sample a \emph{selection key} ($k_\mathsf{blue}$, $k_\mathsf{red}$, $k_\mathsf{green}$), and (2) we encrypt each row's \textsf{bname} with the key corresponding to its \textsf{color}.
In \textsf{ViewGen}, we map each element of $?x$ to its selection key.
For example, for the view $\sqlbg{\sqlselect{}\mathsf{bname}\sqlwhere{}\allowbreak\mathsf{color}\sqlin{}\allowbreak(\mathsf{red}, \mathsf{green})\sqlend{}}$, the view key is $\{k_\mathsf{red}, k_\mathsf{green}\}$.
In \textsf{RevealView}, we use $k_\mathsf{red}$ to decrypt rows 2 and 4 and $k_\mathsf{green}$ to decrypt row 3.

Systems like SiRiUS~\cite{goh2003sirius} and CryptDB~\cite[\S{}4]{popa2011cryptdb} associate encryption keys with filenames and principals, respectively---the core insight in the simplified protocol above is to associate keys with the \textsf{color} field as if it were a filename/principal.
\sys{}'s actual protocol is more complex, as it adds optimizations and levels of indirection to efficiently support multiple view families, multiple $\sqlnsor{}$ predicates in the $\sqlnswhere{}$ clause, etc.

\textbf{Two additional ideas} in \sys{}'s protocol are of particular importance.
First, instead of having a selection key for each value of a field (e.g., \textsf{color}),
we have a selection key for each output of an \emph{arbitrary function} $g$ applied to row contents.
This allows $\sqlnswhere{}$ clauses like $\sqlbg{g(\mathsf{row})\sqlin{}?x}$.
Second, whereas filepaths/principals in CryptDB are public, the \textsf{color} field in the above example is hidden.
Thus, users need a way in \textsf{RevealView} to quickly identify which rows to decrypt and which keys to use.
To achieve this, we use ideas from a network middlebox protocol~\cite{sherry2015blindbox} to develop \emph{key-hiding tags}, which allow data scientists to identify rows to decrypt \emph{up to 50,000$\times$ faster} than na\"ively trying to decrypt each row.

\subsubsection{Orchestrator (\secref{sec:orchestrator})}

The \ul{orchestrator} fetches partitions of a table from storage to compute and invokes \sys{}'s backend on them.
For \textsf{EncryptTable}/\textsf{AddFamily}, it writes the output partitions back to storage.
For \textsf{RevealView}, the output partitions contain decrypted, plaintext data; they are kept at the compute nodes (e.g., in \texttt{tmpfs}).
A data scientist can then load the decrypted data into an analytics framework on those compute nodes.

\subsection{Limitations}

\sys{} supports only a subset of SQL forms in its \sqlviews{} (\secref{s:supported_sql}).
In \secref{sec:evaluation}, we show that the SQL forms \sys{} supports can capture a number of access policies based on realistic use cases.
Further, in \sys{}, the restrictions on SQL apply only to \emph{AC views}, not \emph{analytical queries}.
This is important because analytical queries may be more complex than \sqlviews{} (see \appref{app:analytical_queries}).

Differential privacy (DP)~\cite{dwork2006differentialprivacy} is often applied to aggregates.
Because \sys{}'s views do not support aggregates, they do not support DP.
Still, if the table itself contains aggregates, then the data owner can apply DP before calling \textsf{EncryptTable}.
Similarly, a data scientist who calls \textsf{RevealView} and trains an ML model on the result can apply DP to the model.

\section{\sys{}'s Orchestrator}
\label{sec:orchestrator}

The orchestrator determines the set of partitions to process (\secref{s:linear_scan}) and divides it into \emph{batches}, processed in a streaming fashion.
For each batch, it (1) downloads the partitions in that batch from data lake storage to a compute node's memory, (2) invokes \sys{}'s backend to compute the output partitions using multiple CPU cores, and (3) if needed, writes the output partitions back to storage.
To overlap computation and I/O, the orchestrator \emph{pipelines} the above stages, processing batches in parallel.
Parallel processing of partitions within a batch (\secref{s:architecture}) serves a different role---to use multiple CPU cores.

\parhead{Choosing which partitions to fetch}
\label{s:linear_scan}
For \textsf{RevealView}, the orchestrator may need to fetch only a subset of a table's partitions.
For example, \asqlview{} may be fully contained within a subset of a table's partitions.
Or, a data scientist may only wish to analyze part of her \sqlview{}.
Fetching only some partitions, however, results in \emph{access pattern leakage}---the storage servers learn which rows were or were not fetched.
This is not unique to \sys{}; most EDBs and ESSes leak which rows/documents their clients query.
Sadly, prior research shows that this seemingly innocuous metadata leakage can imply leakage of actual encrypted data ~\cite{naveed2015inference, grubbs2016breaking, grubbs2017why, grubbs2018breaking, cash2015leakage, zhang2016fileinjection}.
Merely partitioning tables with respect to certain columns, as is typical (e.g., so all rows in a partition have the same value for those columns), could leak data via the sizes of partitions.

To reduce such leakage, one can partition tables independently of their contents, and \emph{fetch all partitions} in a table for \textsf{RevealView}.
Because data lakes use \emph{local-area} or \emph{intra-datacenter} networks (e.g., within a cloud region) where bandwidth is plentiful and cheap, the network cost of fetching all partitions is less significant in data lakes than in EDBs/ESSes.

Still, fetching all partitions may be undesirable due to storage I/O costs.
Thus, some users may wish to incur access pattern leakage for better efficiency.
The decision may depend on the data semantics and partitioning scheme.
For example, consider the NYC Taxi Dataset~\cite{yellow-cab-dataset}, partitioned by time; access pattern leakage may be acceptable if partitioned at month granularity, but not if partitioned at minute granularity.

\sys{} allows data scientists/owners/administrators to choose the best option for each application.
Data owners choose how to partition a dataset, and data scientists choose which partitions to fetch via a filter ``\textsf{fil}'' (\figref{fig:api}), a range of partition IDs in our implementation.
If \textsf{fil} is specified, then the orchestrator fetches only matching partitions.

\section{\sys{}'s Planner and Canonical Form}
\label{sec:planner}

\subsection{Supported SQL Forms}
\label{s:supported_sql}
\sys{} supports \sqlnswhere{} clauses with conjunctions (\sqlnsand{}) and disjunctions (\sqlnsor{}) of \emph{predicates}.
Each predicate consists of a field, an operation, and a wildcard.
The field can be \emph{any deterministic function} applied to the contents of a row.
The operation can be $=$, $<$, $\leq$, $>$, $\geq$, or $\neq$ for integer types and $=$ or $\neq$ for strings.
For example, for a table with column names $a$ (string), $b$ (integer), and $c$ (integer), valid predicates are $\sqlbg{\sqllower{a} =\ ?x}$, $\sqlbg{b \geq\ ?x}$, or $\sqlbg{b^2 + c \neq\ ?x}$.

In general, the \sqlnsselect{} clause must include fields used in the \sqlnswhere{} clause.
The reason is that the process of decrypting a row reveals which disjunctive predicates in canonical form are true for that row, leaking information about those fields.
Having those fields in the \sqlnsselect{} clause makes this explicit.

\subsection{\sys{}-Canonical Form}
\label{s:canonical_form}

\sys{}-canonical form is as follows:
\begin{multline}
\sqlbg{\sqlselect{}\mathsf{columns}\sqlfrom{}t\sqlwhere{}} \\
\sqlbg{g_1(\mathsf{row})\sqlin{}?x_1\sqlor{}g_2(\mathsf{row})\sqlin{}?x_2\sqlor{}\ldots\sqlend{}}
\label{query:canonical_form}
\end{multline}
Here, $g_1, g_2, \ldots$ are arbitrary functions.
For example, suppose that $a$ (int), $b$ (int), and $c$ (string) are columns of the table $t$.
Then, the following is in \sys{}-canonical form:
\begin{multline}
\sqlbg{\sqlselect{}a\texttt{, }b\texttt{, }c\sqlfrom{}t\sqlwhere{}}\\
\sqlbg{b\sqlin{}?x_1\sqlor{}\sqllower{c}\sqlin{}?x_2\sqlor{}a\texttt{ + }b\sqlin{}?x_3\sqlend{}}
\label{query:view_family_example}
\end{multline}
Here, $\mathsf{columns}$ is the list ``$a\texttt{, }b\texttt{, }c$,'' $g_1$ is a function that returns field $b$ from a record, $g_2$ is a function that returns field $c$ from a record transformed to lowercase, and $g_3$ is a function that returns the sum of fields $a$ and $b$ from a record.

A canonical-form \viewfamily{} describes \sqlviews{} obtained by replacing wildcards $?x_i$ with \emph{sets} of values.
For example, the view below belongs to the view family in \queryref{query:view_family_example}:
\begin{multline}
\sqlbg{\sqlselect{}a\texttt{, }b\texttt{, }c\sqlfrom{}t\sqlwhere{}}\\
\sqlbg{b\texttt{ = }7\sqlor{}b\texttt{ = }8\sqlor{}\sqllower{c}\texttt{ = "hello"}\sqlend{}}
\label{query:view_example}
\end{multline}
Here, $x_1$ was replaced with $\{7, 8\}$, $x_2$ was replaced with $\{\texttt{"hello"}\}$, and $x_3$ was replaced with $\varnothing$.

\subsection{Rewriting \sqlViews{} into Canonical Form}
\label{s:rewriting}
\label{s:ands}
\label{s:ranges}
\label{s:planner}

To support \sqlnsand{} we use ``secure concatenation''---that is, concatenation that unambiguously delimits the concatenated items.
We denote the secure concatenation of $a$ and $b$ as $a \concat b$.
Our idea is to transform a conjunction of $\sqlnsin{}$ clauses into a single $\sqlnsin{}$ clause of secure concatenations.
For example, $\sqlbg{b\sqlin{}(7)\sqland{}\sqllower{c}\sqlin{}(\texttt{"ab"})}$ becomes $\sqlbg{b \concat \sqllower{c}\sqlin{}(7 \concat \texttt{"ab"})}$. 
The result grows multiplicatively with the sizes of the input clauses.
For example, $\sqlbg{b\sqlin{}(7, 8)\sqland{}\sqllower{c}\sqlin{}(\texttt{"ab"}, \texttt{"cd"})}$ becomes $\sqlbg{b \concat \sqllower{c}\sqlin{}(7 \concat \texttt{"ab"}, 7 \concat \texttt{"cd"}, 8 \concat \texttt{"ab"}, 8 \concat \texttt{"cd"})}$.
We call this as the \emph{combination effect} and measure it in \secref{sec:evaluation}.

\begin{figure}
  \centering
  \includegraphics[width=\linewidth]{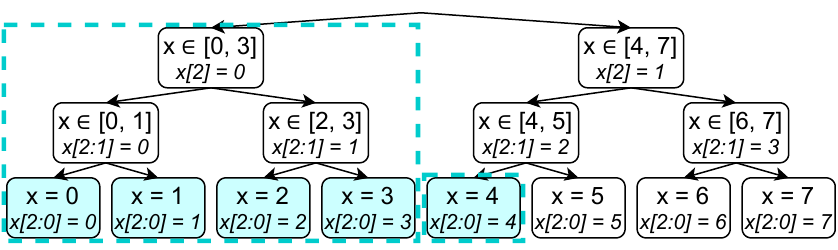}
  \caption{A binary tree covering all 3-bit integers. The range $\sqlbg{0\leq x\leq 4}$ is covered by subtrees at $\sqlbg{x[2] = 0}$ and $\sqlbg{x[2:0] = 4}$.}
  \label{fig:tree}
\end{figure}

To support inequalities, we first observe that inequality predicates (\texttt{<}, \texttt{>}, \texttt{<=}, \texttt{>=}, and \texttt{!=}) can all be transformed into \emph{range} predicates of the form $\sqlbg{a\leq x\leq b}$.
For example, $\sqlbg{x\texttt{ < }a}$ can be transformed into $\sqlbg{\ell \leq x \leq a - 1}$, where $\ell$ is the lowest integer.
Similarly, $\sqlbg{x\texttt{ != }a}$ (equivalently, $\sqlbg{x\sqlnotin{}(a)}$) can be transformed into $\sqlbg{\ell \leq x \leq a - 1\sqlor{}a + 1 \leq x \leq h}$, where $h$ is the highest integer.
Next, we rewrite each range ($\sqlbg{a\leq x\leq b}$) into a disjunction of $\sqlnsin{}$ clauses.
A na\"ive approach is to list each value in the range.
Instead, our approach, inspired by prior bit tree techniques~\cite{kumar2019jedi}, is to imagine a tree over the domain of the integer type.
We represent a range as a list of subtrees \emph{logarithmic in the length of the range}, and represent each subtree as an $\sqlnsin{}$ clause.
\figref{fig:tree} depicts the example $\sqlbg{0 \leq x \leq 4}$ over a 3-bit integer.
The full canonical-form expression for $\sqlbg{0 \leq x \leq 4}$ is $\sqlbg{x[2:0]\sqlin{}(4)\sqlor{}x[2]\sqlin{}(0)}$.

Increasing the branching factor makes the tree shallower; this means fewer predicates in the \viewfamily{}, but more values per wildcard in the \sqlview{}.
We discuss this trade-off in \secref{sec:evaluation}.
We can also support \texttt{!=} for strings by mapping strings to integers (e.g., with collision-resistant hashing).

\section{\sys{}'s Backend}
\label{sec:backend}

\sys{}'s cryptographic protocols have the syntax:
\begin{itemize}[leftmargin=*,topsep=0ex,noitemsep]
    \item $\mathsf{EncryptTable}(t, \tabkey{}, p) \rightarrow t'$
    \item $\mathsf{AddFamily}(t, \tabkey{}, \family{}, \famkey{}, p) \rightarrow t'$
    \item $\mathsf{ViewGen}(\view{}, \famkey{}) \rightarrow \viewkeyset{}$
    \item $\mathsf{RevealView}(t, \viewkeyset{}, p) \rightarrow t'$
\end{itemize}
These are the operations in \figref{fig:api}, with three minor differences.
First, the above protocols operate on a single \emph{partition}, not an entire table; $p$ is the partition ID, counting upwards starting at $1$ (so $p \neq 0$).
Second, while users of \sys{} think of $\viewkey{}$ as a key, it is actually a \emph{set} of keys.
So, in this section, we write $\viewkey{}$ as $\viewkeyset{}$ (capital letter denotes that it is a set).
Third, $\tabkey{}$ and $\famkey{}$ are \emph{inputs}, not \emph{outputs}, as these are \emph{partition-level} operations, and the same table/family key is used for each partition of the table.
For the table-level \textsf{EncryptTable} operation (in \figref{fig:api}), the backend samples $\tabkey{}$ uniformly at random, and then uses that $\tabkey{}$ for all partitions of the table.
For the table-level \textsf{AddFamily} in \figref{fig:api}, the backend accepts $\tabkey{}$ (the table key used to encrypt $t$) as an argument, samples $\famkey{}$ uniformly at random, and then uses those $\tabkey{}$ and $\famkey{}$ for all partitions.

\subsection{Cryptographic Primitives}
\label{s:preliminaries}

A \emph{pseudorandom function} (PRF) is a deterministic function that takes as input a key $k$ and a message $x$.
We denote its application as $\prf{k}{x}$.
To a party who does not know $k$, a PRF's output, for each $x$, appears uniformly random.
The key $k$ is a $\lambda$-bit string; $\lambda = 128$ in our implementation.
We typically assume that the message $x$ and output are also $\lambda$-bit strings, but sometimes allow $x$ to have arbitrary length.

We denote symmetric-key encryption of message $m$ with key $k$ as $\enc{k}{m}$.
The scheme must have two properties.
First, it must be CPA-secure~\cite{boneh2020cpa}: Informally, to any party who does not know $k$, $\enc{k}{m_1}$ and $\enc{k}{m_2}$, for $m_1$ and $m_2$ of equal length, appear to be identically distributed.
Second, it must be key-private~\cite{abadi2000twoviews, bellare2001keyprivacy}: Informally, to any party who does not know $k_1$ and $k_2$, $\enc{k_1}{m}$ and $\enc{k_2}{m}$, for any $m$, appear to be identically distributed.
Our $\mathsf{PRF}$ and $\mathsf{Enc}$ instantiations use the AES block cipher, which has hardware acceleration in commodity x86-64 CPUs via AES-NI.

$\mathsf{OTE}$ (``one-time enc'') denotes an encryption scheme optimized for encrypting only a single value (e.g., one-time pad for small messages).
As described in \secref{s:ands}, $a \concat b$ denotes concatenation of $a$ and $b$ with unambiguous delimitation.

\subsection{Protocol Summary}
\label{s:single_table_architecture}

\sys{}'s backend consists of three layers: \emph{projection}, \emph{selection}, and \emph{tagging}.
\textsf{AddFamily} adds a column for each layer in order; \textsf{RevealView} works in the opposite order.

The \emph{projection layer} handles the $\sqlnsselect{}$ clause.
For \textsf{AddFamily}, this layer adds a \emph{projection column} to the table and computes a \emph{projection key}, $\projkey{r}$, for each row (index $r$).
In \textsf{RevealView}, the projection key for a row is used with the projection column to decrypt the $\sqlnsselect{}$ed fields for that row.

The \emph{selection layer} handles the $\sqlnswhere{}$ clause.
For \textsf{AddFamily}, this layer adds a \emph{selection column} to the table, containing an encryption of $\projkey{r}$ for each row (index $r$).
For each row, it computes a set $\selkeyset{r}$ of \emph{selection keys}.
$\viewkeyset{}$ contains one key per wildcard value $?x_j$ in the view.
Crucially, they are computed such that the rows matching the view are exactly the rows for which $\viewkeyset{} \cap \selkeyset{r} \neq \varnothing$.
In \textsf{RevealView}, the user decrypts the selection column for each matching row using a key in $\viewkeyset{} \cap \selkeyset{r}$, to get $\projkey{r}$.
$\projkey{r}$ is used to decrypt the projection column to get \sqlnsselect{}ed fields in matching rows.

$\viewkeyset{}$ may contain many keys.
The \emph{tagging layer} adds a \emph{tagging column}, used in \textsf{RevealView} to quickly determine which key(s) in $\viewkeyset{}$ are in $\selkeyset{r}$.
The column contains, in each row, a tag computed from each key in $\selkeyset{r}$.
Checking if a key in $\viewkeyset{}$ matches a tag can be far more efficient than na\"ively  trying to decrypt the row using each key in $\viewkeyset{}$.

We provide a \textbf{full protocol description} in \appref{app:full_protocol_description}.
Below, we explain our protocol using the \viewfamily{} in \queryref{query:family_running_example} and the table in \figref{fig:example_relation} as a running example.
\begin{multline}
\sqlbg{\sqlselect{}\mathsf{bname}\texttt{, }\mathsf{color}\sqlwhere{}}\\
\sqlbg{\mathsf{bname}\sqlin{}?x_\predcolor{1}\sqlor{}\mathsf{color}\sqlin{}?x_\predcolor{2}\sqlend{}}
\label{query:family_running_example}
\end{multline}

\begin{figure}[ht]
    \small
    \centering
    \begin{tabular}{c|c|c|c|}
    \multicolumn{1}{c}{$r$} &
    \multicolumn{1}{c}{{\bf bid} ($c = \colcolor{1}$)} & \multicolumn{1}{c}{{\bf bname} ($c = \colcolor{2}$)} & \multicolumn{1}{c}{{\bf color} ($c = \colcolor{3}$)}\\\cline{2-4}
    \rowcolor{1} & 101 & Interlake & blue\\\cline{2-4}
    \rowcolor{2} & 102 & Interlake & red\\\cline{2-4}
    \rowcolor{3} & 103 & Clipper & green\\\cline{2-4}
    \rowcolor{4} & 104 & Marine & red\\\cline{2-4}
    \end{tabular}
    \caption{Relation of boats' IDs, names, and colors~\cite{ramakrishnan2003relational}.}
    \label{fig:example_relation}
\end{figure}

\subsection{The \textsf{EncryptTable} Operation}

\textsf{EncryptTable} encrypts each cell of a table with a separate key, called a \emph{cell key}.
Cell keys are a layer of indirection---\textsf{AddFamily} encrypts the cell keys according to the view family, and \textsf{RevealView} first decrypts the cell keys for cells matching \asqlview{} and then uses the cell keys to decrypt cell data.

The party running \textsf{EncryptTable} (i.e., the data owner) need \emph{not} remember the cell keys.
The data owner only stores $k$ ($\secp{}$ bits) for each table; \textsf{EncryptTable} uses a PRF to derive cell keys from the table key $k$ on the fly.
Concretely, each cell in a partition of $t$ is identified by its row index $r$ within the partition and its column index $c$.
For each row in partition $p$, we derive a \emph{row key}\footnote{Each partition has its own space of keys. To make this explicit, we could have denoted row keys as $k_{p, r}$ instead of $k_r$, denoted cell keys as $k_{p, r, c}$ instead of $k_{r, c}$, and similarly carried an additional $p$ subscript throughout.} from $k$ as $k_r \definedas{} \prf{k}{p \concat r}$.
For each cell in row $r$, we derive a \emph{cell key} from $k_r$ as $k_{r, c} \definedas{} \prf{k_r}{c}$.
Then, we encrypt each cell using its cell key.
The cell keys are not used to encrypt anything else (\sys{} does not allow edits in place), so we use $\mathsf{OTE}$.
See \figref{fig:example_encrypt}.

\begin{figure}[ht]
    \small
    \centering
    \setlength\tabcolsep{4pt}
    \begin{tabular}{c|c|c|c|}
    \multicolumn{1}{c}{$r$} &
    \multicolumn{1}{c}{bid ($c = \colcolor{1}$)} & \multicolumn{1}{c}{bname ($c = \colcolor{2}$)} & \multicolumn{1}{c}{color ($c = \colcolor{3}$)}\\\cline{2-4}
    \rowcolor{1} & $\onetimeenc{k_{\rowcolor{1}, \colcolor{1}}}{\text{101}}$ & $\onetimeenc{k_{\rowcolor{1}, \colcolor{2}}}{\text{Interlake}}$ & $\onetimeenc{k_{\rowcolor{1}, \colcolor{3}}}{\text{blue}}$\\\cline{2-4}
    \rowcolor{2} & $\onetimeenc{k_{\rowcolor{2}, \colcolor{1}}}{\text{102}}$ & $\onetimeenc{k_{\rowcolor{2}, \colcolor{2}}}{\text{Interlake}}$ & $\onetimeenc{k_{\rowcolor{2}, \colcolor{3}}}{\text{red}}$\\\cline{2-4}
    \rowcolor{3} & $\onetimeenc{k_{\rowcolor{3}, \colcolor{1}}}{\text{103}}$ & $\onetimeenc{k_{\rowcolor{3}, \colcolor{2}}}{\text{Clipper}}$ & $\onetimeenc{k_{\rowcolor{3}, \colcolor{3}}}{\text{green}}$\\\cline{2-4}
    \rowcolor{4} & $\onetimeenc{k_{\rowcolor{4}, \colcolor{1}}}{\text{104}}$ & $\onetimeenc{k_{\rowcolor{4}, \colcolor{2}}}{\text{Marine}}$ & $\onetimeenc{k_{\rowcolor{4}, \colcolor{3}}}{\text{red}}$\\\cline{2-4}
    \end{tabular}
    \caption{Result of \textsf{EncryptTable} applied to \figref{fig:example_relation}.}
    \label{fig:example_encrypt}
\end{figure}

An alternative design is to not have row keys, and instead derive cell keys directly from the table key as $k_{r, c} \definedas{} \prf{k}{p \concat r \concat c}$.
Row keys, however, enable a space-saving optimization in the projection layer, as we shall see next.

\subsection{Projection Layer}
\label{s:single_table_projection_layer}

In \textsf{AddFamily}, \sys{} samples a random $\lambda$-bit projection key $\projkey{r}$ for each row (in the general case---see optimizations below).
The projection column contains the cell keys for the columns in the $\sqlnsselect{}$ clause (i.e., in the \texttt{projection\_fields} list), encrypted using $\projkey{r}$.
More formally, if $c_1, \ldots, c_p$ are the indices of $\sqlnsselect{}$ed columns, then the projection column contains $\enc{\projkey{r}}{k_{r, c_1} \concat \ldots \concat k_{r, c_p}}$.
In \textsf{RevealView}, a user with $\projkey{r}$ can decrypt the cell keys in the projection column and use those cell keys to decrypt the fields in row $r$ corresponding to the $\sqlnsselect{}$ed columns.

A user whose view does not include a row $r$ may, in \textsf{RevealView}, arrive at the projection layer with the wrong value for $\projkey{r}$.
Therefore, \textsf{AddFamily} also includes $\enc{\projkey{r}}{0}$ in the projection column.
This lets \textsf{RevealView} identify if $\projkey{r}$ is incorrect and omit row $r$ from the output if so.
See \figref{fig:example_projection}.

\begin{figure}[ht]
    \small
    \centering
    \begin{tabular}{cc|c|}
        $r$ & \multicolumn{1}{c}{projection key} & \multicolumn{1}{c}{projection column (p. col.)}\\\cline{3-3}
        \rowcolor{1} & $\projkey{\rowcolor{1}}$ samp. rand. & $\enc{\projkey{\rowcolor{1}}}{k_{\rowcolor{1}, \colcolor{2}} \concat k_{\rowcolor{1}, \colcolor{3}}} \concat \enc{\projkey{\rowcolor{1}}}{0}$\\\cline{3-3}
        \rowcolor{2} & $\projkey{\rowcolor{2}}$ samp. rand. & $\enc{\projkey{\rowcolor{2}}}{k_{\rowcolor{2}, \colcolor{2}} \concat k_{\rowcolor{2}, \colcolor{3}}} \concat \enc{\projkey{\rowcolor{2}}}{0}$\\\cline{3-3}
        \rowcolor{3} & $\projkey{\rowcolor{3}}$ samp. rand. & $\enc{\projkey{\rowcolor{3}}}{k_{\rowcolor{3}, \colcolor{2}} \concat k_{\rowcolor{3}, \colcolor{3}}} \concat \enc{\projkey{\rowcolor{3}}}{0}$\\\cline{3-3}
        \rowcolor{4} & $\projkey{\rowcolor{4}}$ samp. rand. & $\enc{\projkey{\rowcolor{4}}}{k_{\rowcolor{4}, \colcolor{2}} \concat k_{\rowcolor{4}, \colcolor{3}}} \concat \enc{\projkey{\rowcolor{4}}}{0}$\\\cline{3-3}
    \end{tabular}
    \caption{Projection layer for \viewfamily{} in \queryref{query:family_running_example}.}
    \label{fig:example_projection}
\end{figure}

There are two space-saving optimizations.
If only one column is $\sqlnsselect{}$ed, say $c_1$, then $\projkey{r}$ is set to the cell key for that column, $k_{r, c_1}$.
If all columns are $\sqlnsselect{}$ed ($\sqlbg{\sqlselect{}\mathsf{*}}$), then $\projkey{r}$ is set to the row key, $k_r$.
In these cases, encrypted cell keys are omitted from the projection column, saving space, and $\enc{\projkey{r}}{0}$ is replaced with $\prf{\projkey{r}}{0}$.

\subsection{Selection Layer}
\label{s:single_table_selection_layer}

In a \sys{}-canonical view family, the $\sqlnswhere{}$ clause has $n$ predicates $\sqlnsor{}$ed together.
The $j$th predicate has the form $\sqlbg{g_j(\mathsf{row})\sqlin{}?x_j}$ (see \queryref{query:canonical_form}).
For each row (index $r$) and predicate (index $j$), we choose a \emph{selection key} $\selkey{r, j}$.
In \textsf{AddFamily}, we encrypt each row's projection key once per predicate as $\enc{\prf{\selkey{r, j}}{0}}{\projkey{r}}$ and put the $n$ ciphertexts in the selection column.
For example, see \figref{fig:example_selection}.
The selection keys $\selkey{r, j}$ are chosen such that, if row $r$ satisfies predicate $j$ for a view, then $\selkey{r, j} \in \viewkeyset{}$.
Thus, in \textsf{RevealView}, we can decrypt the ciphertext for that predicate and get $\projkey{r}$.

\begin{figure}[ht]
    \small
    \centering
    \begin{tabular}{c|c|}
        \multicolumn{1}{c}{$r$} & \multicolumn{1}{c}{selection column (s. col.)}\\\cline{2-2}
        \rowcolor{1} & $\enc{\prf{\selkey{\rowcolor{1}, \predcolor{1}}}{0}}{\projkey{\rowcolor{1}}} \concat \enc{\prf{\selkey{\rowcolor{1}, \predcolor{2}}}{0}}{\projkey{\rowcolor{1}}}$\\\cline{2-2}
        \rowcolor{2} & $\enc{\prf{\selkey{\rowcolor{2}, \predcolor{1}}}{0}}{\projkey{\rowcolor{2}}} \concat \enc{\prf{\selkey{\rowcolor{2}, \predcolor{2}}}{0}}{\projkey{\rowcolor{2}}}$\\\cline{2-2}
        \rowcolor{3} & $\enc{\prf{\selkey{\rowcolor{3}, \predcolor{1}}}{0}}{\projkey{\rowcolor{3}}} \concat \enc{\prf{\selkey{\rowcolor{3}, \predcolor{2}}}{0}}{\projkey{\rowcolor{3}}}$\\\cline{2-2}
        \rowcolor{4} & $\enc{\prf{\selkey{\rowcolor{4}, \predcolor{1}}}{0}}{\projkey{\rowcolor{4}}} \concat \enc{\prf{\selkey{\rowcolor{4}, \predcolor{2}}}{0}}{\projkey{\rowcolor{4}}}$\\\cline{2-2}
    \end{tabular}
    \caption{Selection layer for \viewfamily{} in \queryref{query:family_running_example}.}
    \label{fig:example_selection}
\end{figure}

How are selection keys and view keys generated?
For each \viewfamily{}, \sys{} uses a random $\lambda$-bit view family key $\famkey{}$.
For each predicate in the view family (index $j$), we derive a \emph{predicate key} as $\predkey{j} \definedas{} \prf{\famkey{}}{j}$.
In \textsf{AddFamily}, selection keys are derived from the predicate key as $\selkey{r, j} \definedas{} \prf{\predkey{j}}{g_j(\mathsf{row})}$.
(See \secref{s:canonical_form} for an explanation of $g_j$.)
\textsf{ViewGen} generates a view key $\viewkeyset{}$ as follows.
A view assigns a list of wildcard values to each predicate in the view family; for each wildcard value $x$ assigned to predicate $j$, the view key $\viewkeyset{}$ contains the key $\viewkeycomp{j, x} \definedas{} \prf{\predkey{j}}{x}$.
Observe that if a wildcard value $x$ for predicate $j$ equals $g_j(\mathsf{row})$ for a row, then $\selkey{r, j} = \viewkeycomp{j, x}$, allowing the user to decrypt the $j$th ciphertext in the selection column for row $r$ and obtain $\projkey{r}$, as desired.
This works because selection keys and view keys are both derived from $\famkey{}$.
See \figref{fig:example_selection_key}.

\begin{figure}[ht]
    \small
    \centering
    \begin{tabular}{cc}
        $r$ & \multicolumn{1}{c}{selection keys (for \figref{fig:example_selection})}\\
        \rowcolor{1} & $\selkey{\rowcolor{1}, \predcolor{1}} \definedas{} \prf{\predkey{\predcolor{1}}}{\text{Interlake}}$,\:$\selkey{\rowcolor{1}, \predcolor{2}} \definedas{} \prf{\predkey{\predcolor{2}}}{\text{blue}}$\\
        \rowcolor{2} & $\selkey{\rowcolor{2}, \predcolor{1}} \definedas{} \prf{\predkey{\predcolor{1}}}{\text{Interlake}}$,\:$\selkey{\rowcolor{2}, \predcolor{2}} \definedas{} \prf{\predkey{\predcolor{2}}}{\text{red}}$\\
        \rowcolor{3} & $\selkey{\rowcolor{3}, \predcolor{1}} \definedas{} \prf{\predkey{\predcolor{1}}}{\text{Clipper}}$,\:$\selkey{\rowcolor{3}, \predcolor{2}} \definedas{} \prf{\predkey{\predcolor{2}}}{\text{green}}$\\
        \rowcolor{4} & $\selkey{\rowcolor{4}, \predcolor{1}} \definedas{} \prf{\predkey{\predcolor{1}}}{\text{Marine}}$,\:$\selkey{\rowcolor{4}, \predcolor{2}} \definedas{} \prf{\predkey{\predcolor{2}}}{\text{red}}$\\
    \end{tabular}
    \caption{Selection keys for \figref{fig:example_selection}. Note that $\predkey{\predcolor{1}} \definedas{} \prf{\famkey{}}{\predcolor{1}}$ and $\predkey{\predcolor{2}} \definedas{} \prf{\famkey{}}{\predcolor{2}}$.}
    \label{fig:example_selection_key}
\end{figure}

\begin{figure*}[ht]
    \small
    \centering
    \begin{tabular}{cccc}
    $r$ & $\viewkeyset{} \cap \selkeyset{r}$ & decryption flow & $\counter{k}$ after processing row\\
    \rowcolor{1} & $\{\viewkeycomp{\predcolor{1}, \text{Interlake}}\}$ & $\viewkeycomp{\predcolor{1}, \text{Interlake}} \decrightarrow{\text{s. col}} \projkey{\rowcolor{1}} \decrightarrow{\text{p. col.}} k_{\rowcolor{1}, \colcolor{2}},\:k_{\rowcolor{1}, \colcolor{3}} \decrightarrow{\text{cells}} \text{Interlake},\:\text{blue}$ & $\counter{\viewkeycomp{\predcolor{1}, \text{Interlake}}} = 1,\:\counter{\viewkeycomp{\predcolor{2}, \text{red}}} = 0$\\
    \rowcolor{2} & $\{\viewkeycomp{\predcolor{1}, \text{Interlake}},\:\viewkeycomp{\predcolor{2}, \text{red}}\}$ & $\viewkeycomp{\predcolor{1}, \text{Interlake}}\text{ or }\viewkeycomp{\predcolor{2}, \text{red}} \decrightarrow{\text{s. col.}} \projkey{\rowcolor{2}} \decrightarrow{\text{p. col.}} k_{\rowcolor{2}, \colcolor{2}},\:k_{\rowcolor{2}, \colcolor{3}} \decrightarrow{\text{cells}} \text{Interlake},\:\text{red}$ & $\counter{\viewkeycomp{\predcolor{1}, \text{Interlake}}} = 2,\:\counter{\viewkeycomp{\predcolor{2}, \text{red}}} = 1$\\
    \rowcolor{3} & $\varnothing$ & Cannot decrypt s. col.; $\counter{k}$ values and NETs are unchanged & $\counter{\viewkeycomp{\predcolor{1}, \text{Interlake}}} = 2,\:\counter{\viewkeycomp{\predcolor{2}, \text{red}}} = 1$\\
    \rowcolor{4} & $\{\viewkeycomp{\predcolor{2}, \text{red}}\}$ & $\viewkeycomp{\predcolor{2}, \text{red}} \decrightarrow{\text{s. col.}} \projkey{\rowcolor{4}} \decrightarrow{\text{p. col.}} k_{\rowcolor{4}, \colcolor{2}},\:k_{\rowcolor{4}, \colcolor{3}} \decrightarrow{\text{cells}} \text{Marine},\:\text{red}$ & $\counter{\viewkeycomp{\predcolor{1}, \text{Interlake}}} = 2,\:\counter{\viewkeycomp{\predcolor{2}, \text{red}}} = 2$\\
    \end{tabular}
    \caption{\small \textsf{RevealView} for running example with wildcards $x_\predcolor{1} = \{\text{Interlake}\}$ and $x_\predcolor{2} = \{\text{red}\}$. \textsf{ViewGen} outputs $\viewkeyset{} = \{\viewkeycomp{\predcolor{1}, \text{Interlake}},\:\viewkeycomp{\predcolor{2}, \text{red}}\}$, where $\viewkeycomp{\predcolor{1}, \text{Interlake}} = \prf{\prf{\famkey{}}{\predcolor{1}}}{\text{Interlake}}$ and $\viewkeycomp{\predcolor{2}, \text{red}} = \prf{\prf{\famkey{}}{\predcolor{2}}}{\text{red}}$. Both $\counter{k}$ values start at 0.}
    \label{fig:example_revealview}
\end{figure*}

Given that the ciphertexts in the selection column are computed using a key derived from table data, key-privacy of the encryption scheme is crucial for \sys{}'s security.

It may seem tempting to not have predicate keys, and instead derive selection keys and view key members directly from the view family key as $\selkey{r, j} \definedas{} \prf{\famkey{}}{g_j(\mathsf{row})}$ and $\viewkeycomp{j, x} \definedas{} \prf{\famkey{}}{x}$.
This is insecure; it would allow a view key for $\sqlbg{\mathsf{bname} = \mathsf{``blue"}}$ to decrypt row \rowcolor{1}, for example.

An alternate design is to encrypt projection keys with $\selkey{r, j}$, as $\enc{\selkey{r, j}}{\projkey{r}}$.
We prefer encrypting with $\prf{\selkey{r, j}}{0}$, as it enables key-hiding tags (\secref{s:tagging_layer}) to use $\mathsf{Enc}$ as a black box.

\subsection{Tagging Layer}
\label{s:tagging_layer}

So far the user, in \textsf{RevealView}, must try decrypting each row with each key in $\viewkeyset{}$.
This can be slow.
The tagging layer addresses this with a \emph{tagging column} containing, for each row, a \emph{tag} for each selection key.
These tags let the user identify, with high probability, \emph{without any cryptographic operations for those rows}, rows that they cannot decrypt.

The challenge is that tags may leak information.
For example, computing the tag by cryptographically hashing $\selkey{r, j}$ provides the desired functionality, but is insecure---if two rows have the same tag, a user can deduce that they have the same selection key, and therefore the same predicate value.

To solve this, we develop \emph{key-hiding tags}.
We identify BlindBox Detect~\cite{sherry2015blindbox}, a protocol for network middleboxes, as a starting point.
The idea is to generate a \emph{different} tag each time a key is used, by applying a PRF to the key and a \emph{count} of how many previous rows use that key in the same predicate.
Using this in \sys{} requires a \emph{stateful} scan of a table, as \sys{} must maintain a counter for each selection key in \textsf{AddFamily} and for each key in $\viewkeyset{}$ in \textsf{RevealView}.
Alas, this complicates parallel execution, as counters for a row are not known until all previous rows are processed.

Our solution is to generate tags using a \emph{different key} for each partition. 
This way, partitions can be processed in parallel---although counters in different partitions may collide, the tags will appear independent.
Specifically, we generate a \emph{tagging key} as $\tau_{\selkey{r, j}} \definedas{} \prf{\selkey{r, j}}{p}$ for each of the row's selection keys.
For each selection key $\selkey{r, j}$ the corresponding tag is $\prf{\tau_{\selkey{r, j}}}{\counter{\selkey{r, j}}}$, where $\counter{k}$ is the number of previous rows for which $k$ is a selection key.
In $\mathsf{RevealView}$, the user calculates the \emph{next expected tag (NET)} for each key $\viewkeycomp{j, x}$ as $\prf{\tau_{\viewkeycomp{j, x}}}{\counter{\viewkeycomp{j, x}}}$ and maintains the NETs in a data structure with efficient lookup, like a hash set.
For each row, the user checks if a tag in the tagging column matches a NET in the data structure.
If there is a match, then the user decrypts the row; otherwise, the user skips the row without performing cryptographic operations.
Then, for \emph{each} key $k$ in $\viewkeyset{}$ that can decrypt that row, the user increments $\counter{k}$, recalculates $k$'s NET, and updates the data structure.

To save space, we truncate tags, as BlindBox does with \textsf{RS}.
This allows false positives---truncated tags may match where full tags do not---but false positives will be caught at the projection layer, when checking $\enc{\projkey{r}}{0}$.

An alternative design to key-hiding tags could be to adapt an ESS' index structure to \sys{}'s setting.

\figref{fig:example_revealview} shows \textsf{RevealView}, including the tagging layer.

\section{Implementation}
\label{sec:implementation}

We defined \sys{}-canonical form in Protobuf~\cite{protobuf}.
While canonical views, in principle, can have arbitrary predicate functions $g_j$, our implementation supports selecting a field, bits of a field, and concatenations of those results.
This is enough to support equalities and inequalities on fields, but not {\color{blue}\texttt{UPPER}} or {\color{blue}\texttt{LOWER}}.
Not-equal ($\neq$) queries on strings are supported by first hashing strings to integers with SHA-256.

\subsection{\sys{}'s Backend}

We wrote the backend in C++, using AES-128 block cipher as a PRF (more details in \appref{s:formalism_preliminaries}).
For efficiency, our implementation applies it on batches of input blocks, accelerated with AES-NI instructions.
We instantiate \textsf{Enc} using CPA-secure counter-mode encryption; to optimize storage costs, we choose nonces deterministically based on the cell position, while ensuring they are far enough apart to avoid overlap.
For \textsf{OTE}, we use a one-time pad for short inputs, and counter-mode encryption with a zero nonce for longer inputs.

\sys{}'s backend provides a C++ API that can perform \textsf{Encrypt}, \textsf{AddFamily}, \textsf{ViewGen}, and \textsf{RevealView} operations based on the canonical view for a view family or view.
Given a batch of partitions, provided as files (e.g., in an in-memory file system), it can parallelize \textsf{Encrypt}, \textsf{AddFamily}, and \textsf{RevealView} operations by processing partitions on different CPU cores.
We use Apache Arrow~\cite{arrow} to manipulate relations in \sys{}'s backend because of its ability directly interface a wide range of data analytics tools, including Spark~\cite{zaharia2012rdds, armbrust2015sparksql}, Pandas~\cite{pandas}, and DuckDB~\cite{raasveldt2019duckdb}, and various relational file formats like Parquet~\cite{parquet}, ORC~\cite{orc}, and CSV.

We also implemented an optimization to \textsf{AddFamily} that we call the \emph{selection cache}.
In each row, for each predicate, \sys{}'s backend must use the selection key $\selkey{r, j}$ to compute the encryption key $\prf{\selkey{r, j}}{0}$ in the selection column and a tagging key $\tau_{\selkey{r, j}}$ in the tagging column.
Our insight is that the same value often appears in the same column in multiple rows---for example, a \textsf{State} column may have multiple rows containing \textsf{CA}.
Depending on the predicate (e.g., a predicate like $\textsf{State} = ?x$), such repetitions may cause multiple rows to use the same selection key $\selkey{r, j}$ for a predicate.
The selection cache is a mapping from $\selkey{r, j}$ to the prepared AES key schedules for the selection column encryption key and tagging key, to save the work of re-computing these keys when the selection key appears more than once.
The selection cache for a predicate has limited capacity; we use an LRU eviction policy, re-computing the keys on a miss and using the precomputed keys on a hit.
This makes the cache effective when the selection key repeats in nearby rows, without consuming excessive memory for predicates for which the selection key does not repeat (or repeats rarely).

\subsection{\sys{}'s Orchestrator}

We wrote the orchestrator in Python.
It uses cloud-provided tools (e.g., \texttt{azcopy}) to efficiently transfer data between memory and storage (Azure Data Lake Storage).
It invokes \texttt{azcopy} and \sys{}'s backend by spawning them as separate processes.
\sys{}'s backend operates on files; batches are passed between the processes via an \emph{in-memory} file system (e.g., \texttt{/dev/shm}).
This lets the orchestrator use cloud-provided tools (e.g., \texttt{azcopy}) to efficiently transfer data between memory and storage.

\subsection{\sys{}'s Planner}
\label{s:implementation_extra_planner}

We implemented \sys{}'s planner in Rust.
It transforms SQL statements into \sys{}-canonical form via a series of AST transformations, which we describe below.

Immediately after parsing a SQL statement, the AST has $\sqlnsand{}$s, $\sqlnsor{}$s, and $\sqlnsnot{}$s as internal nodes, and inequalities, equalities, and $\sqlnsin$s as leaf nodes.
First, the planner applies De Morgan's Laws to push $\sqlnsnot{}$s down into the leaves.
Second, it transforms all leaves into ranges using the above rules.
Third, as an optimization, it combines multiple ranges on the same field into as few ranges as possible.
Fourth, it converts ranges into $\sqlnsin{}$s using the above rules.
Fifth, it applies the distributed law to transform the AST into disjunctive normal form (DNF).
Sixth, it eliminates $\sqlnsand{}$ internal nodes using the above rules.
DNF is necessary because our technique for $\sqlnsand{}$s only works on $\sqlnsand{}$s where all children are leaves.

The planner optimizes the AST to produce an output with fewer predicates.
The key optimization is to \emph{consolidate} siblings operating on the same field.
For example, while a SQL statement may specify $\sqlbg{x\texttt{ = "hello"}\sqlor{}x\texttt{ = "world"}}$, the planner will consolidate these two clauses into a single predicate $x$ with two wildcard values, \texttt{"hello"} and \texttt{"world"}, instead of na\"ively producing two predicates each with a single wildcard value.
Similarly, siblings in the tree that are range predicates on the same field can be consolidated into a single range predicate with multiple range values.
Thus, conjunctions of inequalities on the \emph{same} field are simplified by the optimizer into a single array of bit selection predicates.
Conjunctions of inequalities on \emph{different} fields cannot be so optimized; each such inequality is converted to a disjunction of bit-selection, and the conjunction of disjunctions is ``multiplied out'' when transforming the AST into DNF.
This results in a canonical form with many predicates.

\section{Evaluation}\label{sec:evaluation}

\begin{table*}[ht]
    \footnotesize
    \centering
    \setlength\tabcolsep{2pt}
    \begin{tabular}{|c|p{3.8in}|p{1.5in}|c|c|}\hline
        Workload & View Family SQL & View Wildcard Values & $s_1$ (\%) & $s_2$ (\%) \\\hline\hline
        \rwestate{} & $\sqlbg{\sqlselect{}\sqlwildcard{}\sqlfrom{}\mathtt{rwe}\sqlwhere{}\mathtt{PATIENT\_STATE}=\ ?x\sqlend{}}$ & [Alaska, California, Hawaii, Oregon, Washington] & 14.16 & 15.17 \\\hline
        \medication{} & $\sqlbg{\sqlselect{}\sqlwildcard{}\sqlfrom{}\mathtt{medications}\sqlwhere{}\mathtt{CODE}=\ ?x\sqlend{}}$ & [212033, 243670, 2563431] & 0.34 & 0.38 \\\hline
        \diagnosis{} & $\sqlbg{\sqlselect{}\sqlwildcard{}\sqlfrom{}\mathtt{conditions}\sqlwhere{}\mathtt{CODE}=\ ?x\sqlend{}}$ & 414545008 & 0.32 & 0.32\\\hline
        \rweobsstate{} & 
        $\sqlbg{\sqlselect{}\mathtt{OBSERVATION\_DATE}, \mathtt{OBSERVATION\_CATEGORY}, \mathtt{OBSERVATION\_CODE}}$ $\sqlbg{\mathtt{OBSERVATION\_DESCRIPTION}, \mathtt{OBSERVATION\_VALUE}, \mathtt{OBSERVATION\_UNITS}}$ $\sqlbg{\mathtt{OBSERVATION\_TYPE}, \mathtt{PATIENT\_STATE}\sqlfrom{}\mathtt{rwe}\sqlwhere{}\mathtt{PATIENT\_STATE}=\ ?x\sqlend{}}$ & [Alaska, California, Hawaii, Oregon, Washington] & 14.16 & 15.17 \\\hline
        \dropoffpickup{} & $\sqlbg{\sqlselect\sqlwildcard{}\sqlfrom{}\mathtt{yellow}\sqlwhere{}\mathtt{doLocationId}=\ ?x \sqland{} \mathtt{puLocationId} =\ ?y\sqlend{}}$ & [(236, 236), (236, 237), (237, 236), (237, 237)] & 2.12 & 0.40\\\hline
        \salesstore{} & $\sqlbg{\sqlselect{}\sqlwildcard{}\sqlfrom{}\mathtt{store\_sales}\sqlwhere{}\mathtt{ss\_store\_sk} =\ ?x\sqlend{}}$ & 7 & 15.80 & 0.14 \\\hline
        \salesdate{} & $\sqlbg{\sqlselect{}\sqlwildcard{}\sqlfrom{}\mathtt{store\_sales}\sqlwhere{}\mathtt{ss\_sold\_date\_sk} >\ ?x\sqlend{}}$ $\sqlbg{\sqland{} \mathtt{ss\_sold\_date\_sk} <\ ?y}$ & (2452411, 2452642) & 14.89 & 14.94 \\\hline\hline
        \rweineqobs{} & $\sqlbg{\sqlselect{}\sqlwildcard{}\sqlfrom{}\mathtt{rwe}\sqlwhere{}\mathtt{OBSERVATION\_DATE}\geq\ ?x\sqlend{}}$ & "2022-01-01 00:00:00+00:00" & 14.64 & 14.17 \\\hline
        \rweineqstate{} & $\sqlbg{\sqlselect{}\sqlwildcard{}\sqlfrom{}\mathtt{rwe}\sqlwhere{}\mathtt{PATIENT\_STATE}\neq\ ?x\sqlend{}}$ & California & 90.51 & 89.06 \\\hline
        \rweineqor{} & $\sqlbg{\sqlselect{}\sqlwildcard{}\sqlfrom{}\mathtt{rwe}\sqlwhere{}\mathtt{PATIENT\_DEATHDATE}\neq\ ?x\sqlor{}}$ $\sqlbg{\mathtt{OBSERVATION\_DATE}\geq\ ?y\sqlend{}}$ & (NULL, \newline "2022-01-01T00:00:00+00:00") & 43.05 & 47.30 \\\hline
        \rweineqand{} & $\sqlbg{\sqlselect{}\sqlwildcard{}\sqlfrom{}\mathtt{rwe}\sqlwhere{}\mathtt{PATIENT\_DEATHDATE}\neq\ ?x\sqland{}}$ $\sqlbg{\mathtt{OBSERVATION\_DATE}\geq\ ?y\sqlend{}}$ & (NULL, \newline "2022-01-01T00:00:00+00:00") & 0.43 & 0.26 \\\hline
    \end{tabular}
    \caption{Workloads to evaluate \sys{}. Selectivity for the ``small'' version of the datasets is $s_1$, and selectivity for the ``large'' version is $s_2$.}
    \label{tab:workloads}
\end{table*}

We use three datasets: 
(1) a synthetic medical dataset generated using Synthea~\cite{synthea-dataset}, (2) New York City yellow taxi trip records ~\cite{yellow-cab-dataset}, and (3) LHBench, a TPC-DS-based dataset for benchmarking data lakes~\cite{jain2023lhbench}. 
We use two versions of each---a small version for testing single-core performance (max table size $\approx 2$ GiB, uncompressed), and a large version to test scalability.
In the large version, \rwe{} (``real-world medical evidence'' table created using Synthea\footnote{Synthea does not directly output \rwe{}. To produce an \rwe{} table similar to Eisai's demo, we join the \textsf{Patient}, \textsf{Observation}, and \textsf{Encounter} tables. All Synthea tables were produced by running Synthea separately for each US State, with sizes proportional to their populations, and combining the data.}) was $\approx 200$ GiB uncompressed, \yellow{} (table from NYC dataset) was $\approx 250$ GiB uncompressed, and \storesales{} (table from LHBench) was $\approx 1$ TiB uncompressed.
In the large version, \textsf{medications} and \textsf{conditions} tables (from Synthea) were only tens of GiB.

\parhead{Views based on realistic use cases}
The first \sqlview{}, \rwestate{}, is the Eisai example~\cite{esaisummittalk2021} from \secref{sec:intro}, granting access to rows from \rwe{} for certain US States.
The Eisai demo also discusses the importance of hiding entire columns; this inspires \rweobsstate{}, which is the same as \rwestate{} except that it only $\sqlnsselect{}$s 9 columns.
SMCQL~\cite{bater2017smcql} describes an analytical query counting patients prescribed aspirin and then diagnosed with heart disease.
This inspires our next two \sqlviews{}: \medication{} grants access to data for certain medication codes (aspirin) in Synthea's \textsf{medications} table, and \diagnosis{} grants access to diagnoses of certain condition codes (heart disease) in Synthea's \textsf{conditions} table.
Location privacy is an important concern for datasets like NYC Taxi Cab~\cite{lecuyer2019sage}, inspiring our next \sqlview{}: \dropoffpickup{} grants access to taxi rides for particular combinations of dropoff and pickup zones, using \sys{}'s support for \sqlnsand{}.
LHBench includes sales and customer data.
Like the Eisai example, \salesstore{} grants access to data in \storesales{} for a particular set of stores.
As historical data is often restricted, \salesdate{} grants access to data in \storesales{} for a particular time range.

\parhead{Views to stress \sys{}}
These \sqlviews{} all use \rwe{}; we vary only the SQL.
\rweineqobs{} grants access to patient observations in a time range.
\rweineqstate{} does an inequality check on \emph{strings}, requiring hashing strings to integers and forming a tall tree over 256-bit integers.
\rweineqor{} is an \sqlnsor{} of the inequality in \rweineqobs{} and an inequality checking that a patient's death date is not \textsf{NULL}.
\rweineqand{} is like \rweineqor{}, but is an \sqlnsand{} of the inequalities instead of an an \sqlnsor{}.
This requires ``multiplying out'' the conjunction into DNF (\secref{s:planner}), creating many predicates.

SQL for these views is in \tabref{tab:workloads}.

\subsection{\sys{}'s Cryptographic Protocol}
\label{s:eval_protocol}

\begin{figure}[t]
    \centering
    \begin{subfigure}[t]{0.31\linewidth}
        \centering
        \includegraphics[width=\linewidth, trim=6pt 0pt 6pt 0pt, clip]{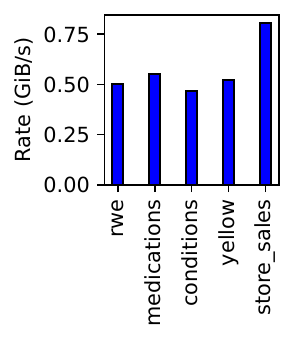}
    \end{subfigure}
    \begin{subfigure}[t]{0.68\linewidth}
        \centering
        \includegraphics[width=\linewidth, trim=6pt 0pt 6pt 0pt, clip]{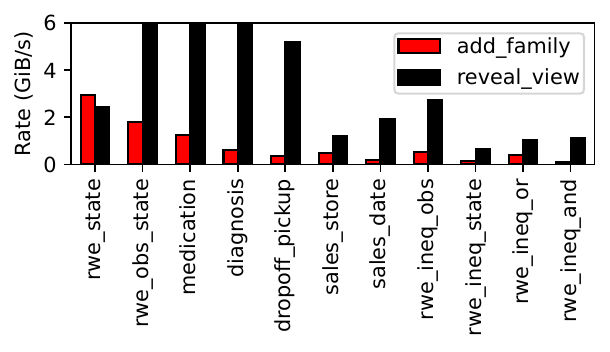}
    \end{subfigure}
    \caption{Throughput of \sys{}'s backend on in-memory data.}
    \label{fig:rate}
\end{figure}

We measure \sys{}'s backend's performance on a single core using the small version of the datasets.
To measure backend performance, we divide the time to process each table in memory (excluding reading/writing the input/output) by the \emph{uncompressed} plaintext size (\figref{fig:rate}).
\textbf{\sys{}'s backend runs at hundreds of megabytes to gigabytes per second \emph{on a single core}}, showing that our design based on hardware-accelerated, symmetric-key cryptography is performant.

\textsf{EncryptTable} is fastest for \textsf{store\_sales} because all of its cells are at most 16 bytes, so \textsf{OTE} can use the fast one-time pad for all cells.
\textsf{AddFamily} is faster for many-column tables (e.g., \rwe{}) because \textsf{AddFamily} only computes on columns in the \sqlnswhere{} clause (a small fraction of a many-column table).

\textsf{RevealView} is much faster for \sqlviews{} that match fewer rows (i.e., have low selectivity) because, with key-hiding tags (\secref{s:tagging_layer}), it performs no cryptography for non-matching rows.
\figref{fig:selectivity_rwe_state} varies the State in the \sqlnswhere{} clause for \rwestate{}, showing that \textsf{RevealView} performance is linear in selectivity.

\Viewfamilies{} whose canonical form has \emph{many predicates} are generally slower to process.
This affects inequalities, which are rewritten to many predicates (\secref{s:ranges}), and particularly conjunctions of ranges, which must be ``multiplied out'' to DNF.
For example, \rweineqand{} has low \textsf{AddFamily} throughput, and similar \textsf{RevealView} throughput as \rweineqor{}, despite matching fewer rows.

\textsf{ViewGen} (not graphed) is most expensive for many-predicate \sqlviews{}.
For a branching factor of 256 (our default) for inequality trees (\secref{s:ranges}), its overheads are modest: $\approx 84$ ms runtime and $\approx 438$ MB memory for \rweineqand{}, and $\approx 30$ ms runtime and $\approx 31$ MB memory for \rweineqor{}.

\rweobsstate{} \sqlnsselect{}s only some columns.
This speeds up \textsf{RevealView}, as \textsf{RevealView} decrypts \sqlnsselect{}ed columns only.
It slows \textsf{AddFamily} because it precludes the optimization that omits encrypted cell keys and takes $\projkey{r} = k_r$ (\secref{s:single_table_projection_layer}).

\rweineqstate{} involves string inequality.
This slows \textsf{AddFamily} because it computes a SHA-256 hash per row.

\begin{figure}[t]
    \centering
    \begin{minipage}{0.43\linewidth}
        \includegraphics[width=1.1\linewidth]{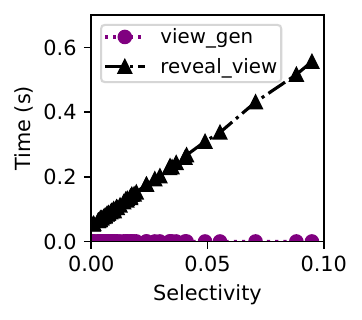}
        \caption{Selectivity.}
        \label{fig:selectivity_rwe_state}
    \end{minipage}
    \begin{minipage}{0.56\linewidth}
        \centering
        \includegraphics[width=1.02\linewidth, trim=0pt 8pt 0pt 0pt, clip]{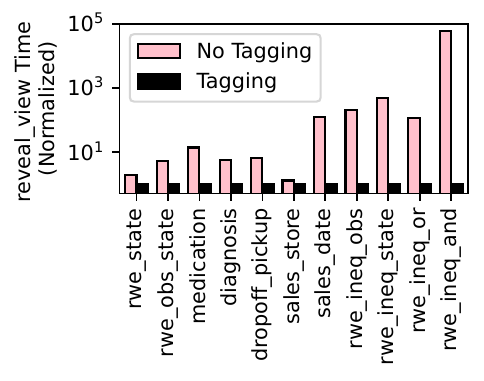}
        \caption{Tagging in \textsf{RevealView}.}
        \label{fig:features_revealview}
    \end{minipage}
\end{figure}

\subsection{Key-Hiding Tags and Selection Cache}
As shown in \figref{fig:features_revealview}, \textbf{key-hiding tags speed up \textsf{RevealView} by 1.3$\times$ to over 50,000$\times$}.
For inequality-based views in particular, key-hiding tags are essential to achieving ``big data'' speeds.
This is because the planner rewrites inequalities into disjunctions of many predicates; for a branching factor of~256, $\viewkeyset{}$ can contain \emph{hundreds} of keys per predicate.
With key-hiding tags, the client uses a hash table to quickly find the key to use in $\viewkeyset{}$; without them the client must try decrypting each row with \emph{each key in $\viewkeyset{}$}, which is slow when $\viewkeyset{}$ is large.
The gains are especially significant for view containing inequalities.
\rweineqand{} shows an extreme performance gain for two reasons.
First, conjunctions of inequalities are ``multiplied out'' to DNF, and $\viewkeyset{}$ also grows multiplicatively due to the \emph{combination effect} (\secref{s:planner}).
Second, its selectivity is $<1\%$, so \textsf{RevealView}, with key-hiding tags, handles $>99\%$ of rows without any cryptographic operations.

\begin{figure}[t]
    \centering
    \includegraphics[width=\linewidth]
    {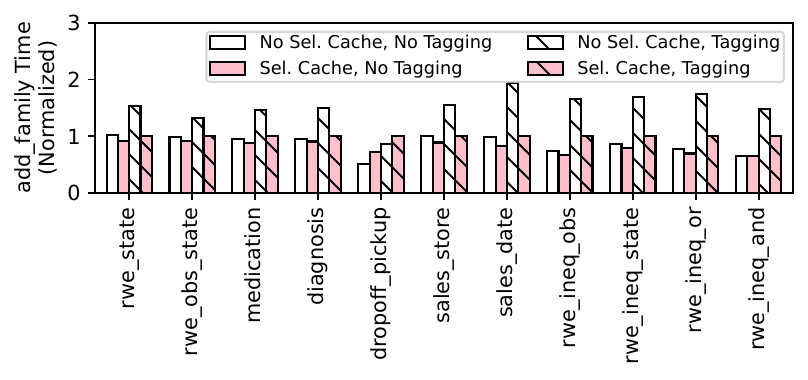}
    \caption{Feature impact (\textsf{AddFamily}).}
    \label{fig:features_addfamily}
\end{figure}

\begin{figure}[t]
\centering
    \includegraphics[width=0.6\linewidth]{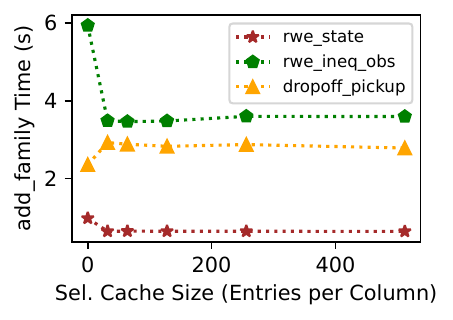}
    \caption{Selection cache size.}
    \label{fig:selcache_time_addfamily}
\end{figure}

\figref{fig:features_addfamily} shows how tagging and the selection cache impact \textsf{AddFamily} performance.
The selection cache mitigates the tagging overhead in \textsf{AddFamily} because it caches the generated key schedule for $\tau_{\selkey{r, j}}$, reducing the overhead of tag generation.
Without the selection cache, computing tags increases \textsf{AddFamily} latency by up to $2\times$.

The selection cache can bring performance gains even at small sizes (\figref{fig:selcache_time_addfamily}).
The reason is that, even if not all unique values of the predicate can fit in the selection cache simultaneously, datasets may be distributed in a way such that only a few unique values constitute most of the occurrences.
Our LRU replacement policy results in the most commonly occurring values usually being represented in the cache.

The best selection cache size varies depending on the data distribution.
For example, at our default selection cache size of 512, the selection cache actually reduces \textsf{AddFamily} performance for the \dropoffpickup{} workload.
At a larger selection cache size of 8192, however, performance gains are realized (not shown in \figref{fig:selcache_time_addfamily}).

\subsection{Compression and Size Overheads}
\label{s:compression_size_overheads}

\begin{figure}[t]
    \centering
    \begin{subfigure}[t]{0.415\linewidth}
        \centering
        \includegraphics[width=\linewidth]{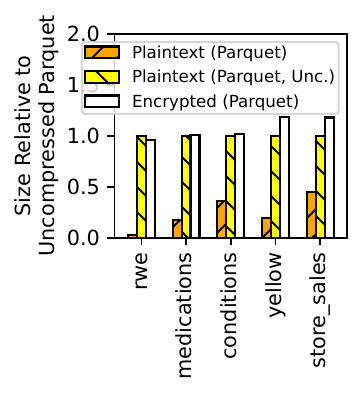}
        \caption{Output of \textsf{EncryptTable}.}
        \label{fig:size_encrypt_normalized}
    \end{subfigure}
    \hspace{-2ex}
    \begin{subfigure}[t]{0.60\linewidth}
        \centering
        \includegraphics[width=\linewidth, trim=0pt 8pt 0pt 0pt, clip]{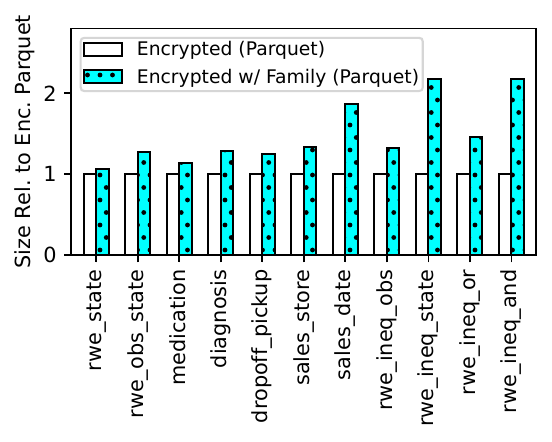}
        \caption{Output of \textsf{AddFamily}.}
        \label{fig:size_addfamily_normalized}
    \end{subfigure}
    \caption{\sys{}'s size overhead.}
    \label{fig:size_normalized}
\end{figure}

\begin{figure}[t]
    \centering
    \includegraphics[width=0.49\linewidth]{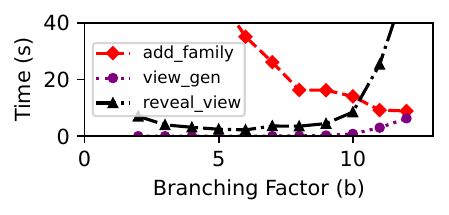}
    \includegraphics[width=0.49\linewidth]{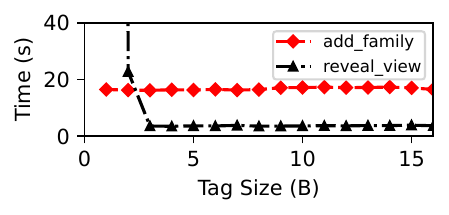}
    \includegraphics[width=0.49\linewidth]{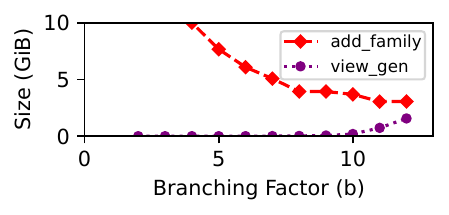}
    \includegraphics[width=0.49\linewidth]{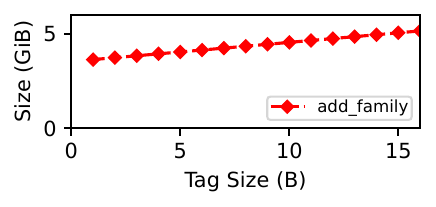}\label{fig:tag_size_rwe_ineq_and}
    \caption{Varying branching factor and tag size, \rweineqand{}. Branching factors are in \emph{bits}; $b=8$ is a branching factor of $2^8 = 256$.}
    \label{fig:trade_offs_rwe_ineq_and}
\end{figure}

\begin{figure}[t]
    \centering
    \includegraphics[width=0.49\linewidth]{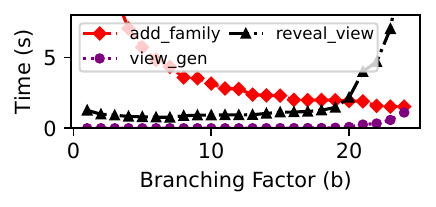}
    \includegraphics[width=0.49\linewidth]{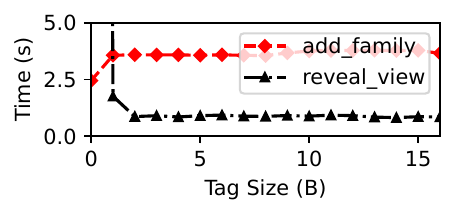}
    \includegraphics[width=0.49\linewidth]{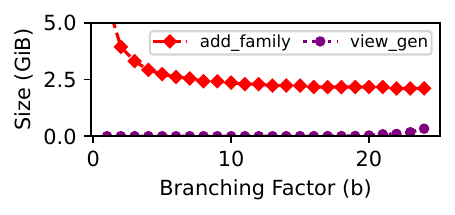}
    \includegraphics[width=0.49\linewidth]{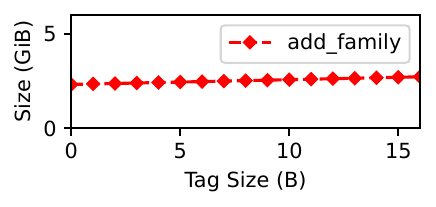}\label{fig:tag_size_rwe_ineq_obs}
    \caption{Varying branching factor and tag size, \rweineqobs{}. Branching factors are in \emph{bits}; $b=8$ is a branching factor of $2^8 = 256$.}
    \label{fig:trade_offs_rwe_ineq_obs}
\end{figure}

We always encrypt data in \emph{uncompressed} form; because encrypted data cannot be compressed, \sys{}'s encryption leads to size overhead.
The output of \textsf{EncryptTable} is similar in size to the \emph{uncompressed} plaintext data, but an order of magnitude bigger than the \emph{compressed} Parquet input (\figref{fig:size_encrypt_normalized}).
The size overhead of \textsf{AddFamily} relative to \textsf{EncryptTable} (due to proj./sel./tagging columns) is typically $\approx$~$1.5\times$ or less, but up to $2\times$ (\figref{fig:size_addfamily_normalized}).
Even including \textsf{AddFamily}'s overheads, loss of compression dominates size overheads.

Our implementation uses compressed Parquet to encode plaintext data (input to \textsf{EncryptTable} and output of \textsf{RevealView}).
For encrypted data (which do not benefit from compression), we use Arrow's serialization format (\textsf{ipc}), which does not have compression but is faster to (de)serialize than Parquet.
To show the benefit of using \textsf{ipc}, \figref{fig:format_addfamily_revealview} compares three serialization formats (Parquet with compression, Parquet without compression, and \textsf{ipc}), with the input and output stored on the local SSD in the same format.
Note that the time to read/write the file (including serialization) is significant compared to the cryptographic processing time, and that \textsf{ipc} tends to be faster than Parquet.

\begin{figure}[t]
    \begin{subfigure}[t]{\linewidth}
        \centering
        \includegraphics[width=0.7\linewidth]{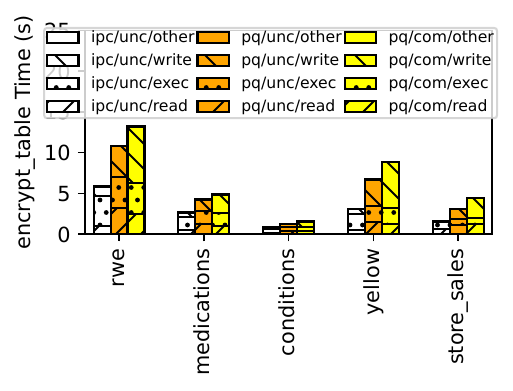}
        \caption{\textsf{EncryptTable}.}
        \label{fig:format_encrypttable}
    \end{subfigure}
    \begin{subfigure}[t]{\linewidth}
        \centering
        \includegraphics[width=\linewidth]{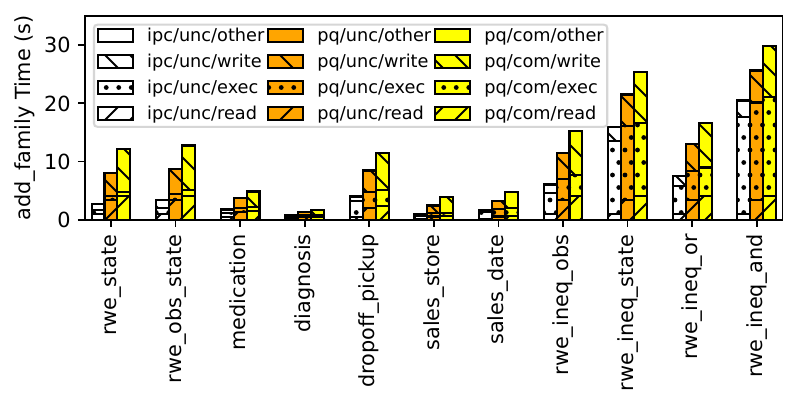}
        \caption{\textsf{AddFamily}.}
        \label{fig:format_addfamily}
    \end{subfigure}
    \begin{subfigure}[t]{\linewidth}
        \centering
        \includegraphics[width=\linewidth]{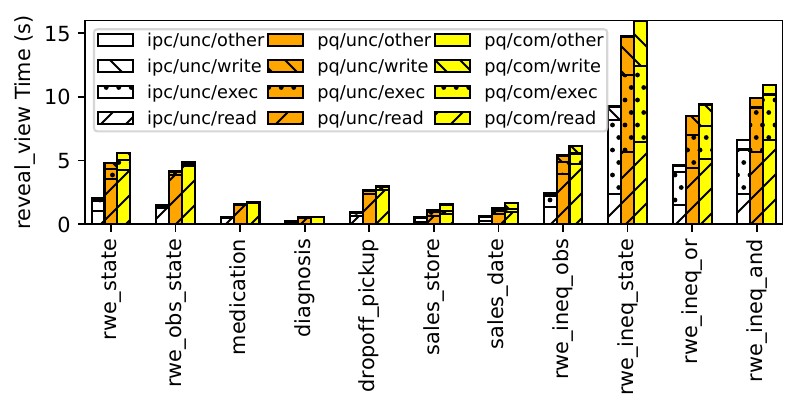}
        \caption{\textsf{RevealView}.}
        \label{fig:format_revealview}
    \end{subfigure}
    \caption{\sys{} total runtime in single-core setting, including reading/writing files from local temp SSD, for different data formats.}
    \label{fig:format_addfamily_revealview}
\end{figure}

\subsection{Space/Time Trade-Offs}

Two factors present a space/time trade-off: (1) truncating tags and (2) choosing the branching factor.
\figref{fig:trade_offs_rwe_ineq_and} and \figref{fig:trade_offs_rwe_ineq_obs} measure these trade-offs for \textsf{rwe\_ineq\_obs} and \textsf{rwe\_ineq\_and}.
We choose these workloads because they use inequalities and have canonical forms with many predicates, so they are most sensitive to the branching factor and tag length.

Decreasing the tag length makes the output of \textsf{AddFamily} smaller (because tags are smaller), but decreasing it too much makes \textsf{RevealView} slower due to frequent false-positive tag matches.
Increasing the branching factor makes the \textsf{AddFamily} faster and its output smaller because it reduces the number of predicates in canonical form, but increasing it too much makes \textsf{RevealView} much slower since it must check more keys for each row.
Importantly, intermediate values (e.g., branching factor of $2^8 = 256$, 4-byte tags) appropriately balance this trade-off and result in all-around good performance.

\subsection{Scalability to Multiple CPU Cores}
\label{s:eval_multicore}

\begin{figure*}[t]
    \begin{minipage}[b]{0.8\linewidth}
        \begin{subfigure}[t]{0.33\linewidth}
            \centering
            \includegraphics[width=0.92\linewidth, trim=0pt 4pt 0pt 0pt, clip]{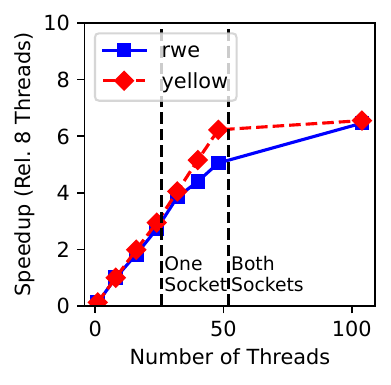}
            \caption{\textsf{EncryptTable}.}
            \label{fig:speedup_multiprocess_encrypttable}
        \end{subfigure}
        \begin{subfigure}[t]{0.33\linewidth}
            \centering
            \includegraphics[width=\linewidth, trim=0pt 4pt 0pt 0pt, clip]{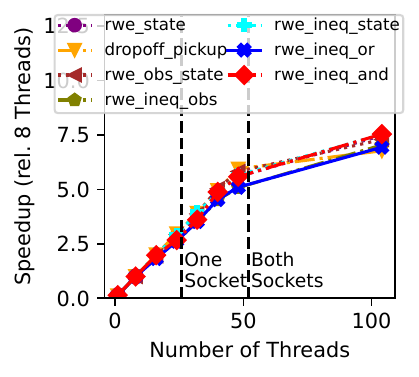}
            \caption{\textsf{AddFamily}.}
            \label{fig:speedup_multiprocess_addfamily}
        \end{subfigure}
        \begin{subfigure}[t]{0.33\linewidth}
            \centering
            \includegraphics[width=\linewidth, trim=0pt 4pt 0pt 0pt, clip]{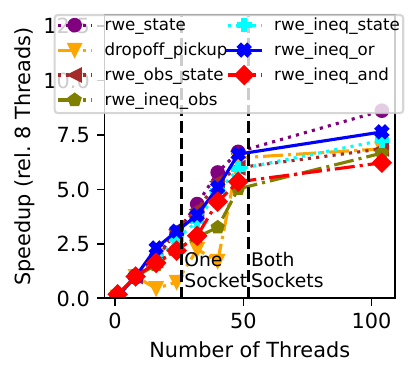}
            \caption{\textsf{RevealView}.}
            \label{fig:speedup_multiprocess_revealview}
        \end{subfigure}
        \caption{Multi-core scalability. Speedup is relative to 8 threads (e.g., linear scaling at 48 cores would be $6\times$).}
        \label{fig:multiprocess}
    \end{minipage}
    \begin{minipage}[b]{0.18\linewidth}
        \centering
        \includegraphics[width=0.9\linewidth, trim=0pt 8pt 2pt 2pt, clip]{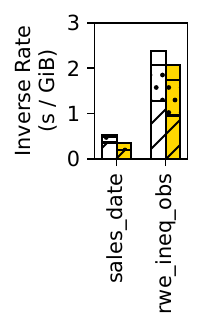}
        \caption{Effect of \textsf{fil}.}
        \label{fig:cloud_filter_revealview}
    \end{minipage}
\end{figure*}

We scale \sys{}'s backend to multiple cores using large, multi-partition datasets.
We use a \texttt{Standard\_E104ids\_v5} instance and omit the \storesales{} table, which does not fit in memory.
As in \secref{s:eval_protocol}, we measure the time for \sys{}'s backend to process in-memory deserialized data, and exclude the time to read the input or write the output.
Because \sys{} uses a separate \emph{process} per core, partitions are \emph{statically} assigned to threads/cores, for this benchmark only.

See \figref{fig:multiprocess}.
\sys{}'s protocols scale linearly across multiple physical cores and sockets.
While our system had 52 cores, we also measure performance with 104 threads, to use logical CPU cores (hyperthreading); as expected their benefit is less than physical cores.
Results for \textsf{RevealView} are noisier than for \textsf{AddFamily}, possibly due to static partitioning, variable runtime for partitions, and NUMA effects.

\subsection{End-to-End Performance}
\label{sec:eval_e2e}

\begin{figure}[t]
    \centering

    \begin{minipage}[b]{\linewidth}
        \centering
        \includegraphics[width=\linewidth, trim=0pt 8pt 0pt 0pt, clip]{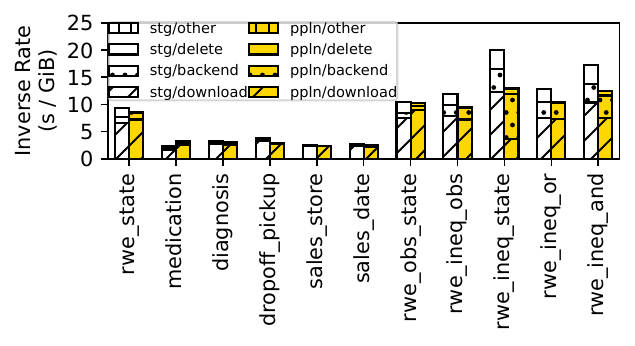}
        \caption{End-to-end breakdown for \textsf{RevealView}.}
        \label{fig:cloud_revealview}
    \end{minipage}
\end{figure}

\begin{figure*}[t!]
    \begin{subfigure}{0.38\linewidth}
        \centering
        \includegraphics[width=\linewidth]{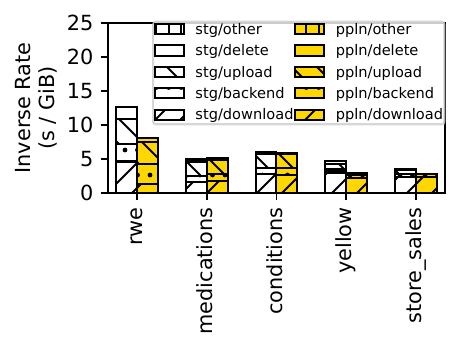}
        \caption{\textsf{EncryptTable}.}
        \label{fig:cloud_encrypttable}
    \end{subfigure}
    \begin{subfigure}{0.61\linewidth}
        \centering
        \includegraphics[width=\linewidth]{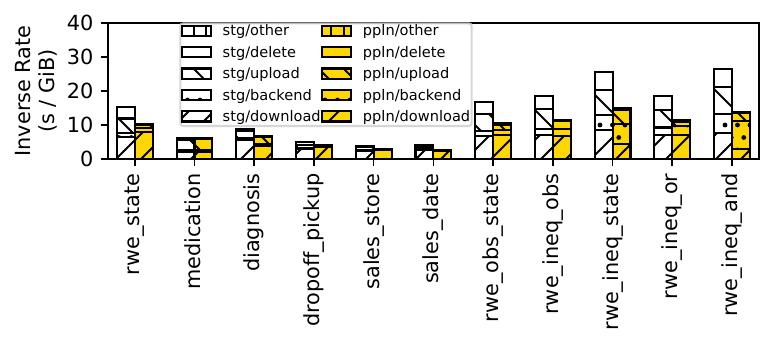}
        \caption{\textsf{AddFamily}.}
        \label{fig:cloud_addfamily}
    \end{subfigure}
    \caption{End-to-end breakdown of total runtime.}
    \label{fig:cloud_extra}
\end{figure*}

\begin{figure}[t]
    \centering
    \begin{subfigure}[t]{0.49\linewidth}
        \centering
        \includegraphics[width=\linewidth, trim=0pt 8pt 0pt 0pt, clip]{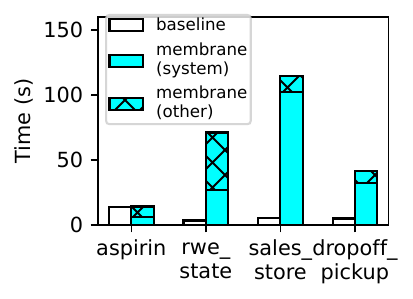}
        \caption{First query in a session.}
        \label{fig:databricks_first_query}
    \end{subfigure}
    \begin{subfigure}[t]{0.49\linewidth}
        \centering
        \includegraphics[width=\linewidth, trim=0pt 8pt 0pt 0pt, clip]{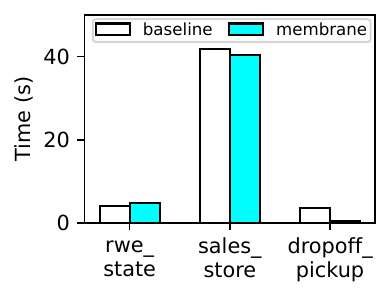}
        \caption{Additional queries.}
        \label{fig:databricks_additional_queries}
    \end{subfigure}
    \caption{\sys{}'s overhead in an interactive session.}
    \label{fig:databricks}
\end{figure}

\parhead{Planning}
With a branching factor of 256, planning typically takes less than 10 ms---the longest time is for \rweineqand{}, at $\approx 300$ ms---and consumes less than 100 MiB of memory.

\parhead{Interacting with cloud storage}
We now measure \textsf{RevealView} when using \sys{}'s orchestrator, both with (``ppln'') and without (``stg'') the orchestrator's pipelining.
We use \texttt{Standard\_E104ids\_v5} and materialize decrypted views on the local SSD.
\figref{fig:cloud_revealview} shows results when processing \emph{all} partitions of a table to hide access patterns (\secref{s:linear_scan}).
While pipelining is generally faster, disabling it shows a clearer breakdown into individual components.
\sys{}'s backend is a significant fraction of overall runtime for some workloads (e.g., \rweineqstate{}), but fetching partitions from storage usually dominates.
\figref{fig:cloud_extra} shows results for \textsf{EncryptTable} and \textsf{AddFamily}---compared to \figref{fig:cloud_revealview}, this graph has an additional component, \emph{upload time}, since \textsf{EncryptTable} and \textsf{AddFamily} upload their output to storage.
\figref{fig:cloud_filter_revealview} shows results when using \textsf{fil} to only process partitions that contain view contents.
We use \rweineqobs{} and \salesdate{}, as their data are in a contiguous range of partitions due to how the tables are sorted.
Using \textsf{fil} is much faster (compare y-axes of \figref{fig:cloud_revealview} \& \figref{fig:cloud_filter_revealview}) and reduces the relative overhead of fetching partitions, at the cost of revealing access patterns.
Both figures show time normalized by the full \emph{compressed} table size.

\parhead{Interactive analytics}
We ran \sys{} in a Databricks-hosted Spark cluster in an interactive notebook.
The cluster has 16 machines, each with 8 CPUs and 64 GiB RAM, similar to LHBench's evaluation setup~\cite{jain2023lhbench}.
Our baseline is to run PySpark-SQL in the standard way, providing it with the dataset's cloud storage URL.
To run \sys{}, we call \texttt{map} on Spark RDDs to run the orchestrator to have workers process disjoint sets of partitions, and use internal Spark APIs to convert output partition files at workers into a Spark dataframe.

For each dataset (Synthea, NYC Taxi, and LHBench), we obtained analytical SQL queries (\appref{app:analytical_queries}).
We separately measure the time to the result of the first SQL query, and the time to run the remaining SQL queries.
After decrypting \asqlview{} with \sys{}, we call \texttt{persist} so that Spark caches it in memory for subsequent queries.
We did not call \texttt{persist} in the baseline, so Spark's query optimizer can co-optimize the \sqlview{} and query.

The results are in \figref{fig:databricks}.
On the first query, the baseline fetches plaintext data from cloud storage.
\sys{} fetches the result of \textsf{AddFamily}, which is larger due to loss of compression, and then runs \textsf{RevealView} on the data.
Thus, \sys{} increases the time until the first query result by up to $20\times$.
Subsequent analytical queries execute on the cached result of \textsf{RevealView} and perform similarly to the baseline.
For \dropoffpickup{}, they were actually \emph{faster} with \sys{}.
This may be because \sys{} locally materializes the view.

\subsection{Comparison to Other Systems}

There is no existing cryptographic system with the same functionality as \sys{}.
The closest is CryptDB's multi-principal design~\cite[\S{}4]{popa2011cryptdb}, but (1) the design does not support \emph{private, encrypted} access control attributes, and (2) the public CryptDB code does not even support multiple principals.
Thus, we are forced to consider baselines that do not exactly match \sys{}'s functionality and/or security.

\parhead{Existing Cryptographic Schemes}
First, we consider Inner Product Encryption (IPE)~\cite{katz2008ipe}, capable of conjunctions and disjunctions, to encrypt each row.
IPE removes the restriction that the \sqlnsselect{} clause must include fields used in the \sqlnswhere{} clause (\secref{s:supported_sql}), as IPE decryption hides which $\sqlnsor{}$ clauses match.
Second, we consider Identity-Based Encryption (IBE)~\cite{boneh2001ibe}; an ID-hiding IBE scheme can be used instead of \sys{}'s selection keys for each predicate.
These designs still rely on \sys{}'s view families, canonical form, and planner.

\begin{table}[t]
    \small
    \centering
    \begin{tabular}{|c|c|c|} \hline
        Protocol & Dec. Latency & Dec. Thrpt (rows/s/core)\\ \hline\hline
        IPE & $\approx 4$ ms & $\approx 250$ \\ \hline
        IBE & $\approx 1$ ms & $\approx 1,000$ \\ \hline
        \sys{} & $\approx 0.002$ ms & $\approx 500,000$ \\ \hline
    \end{tabular}
    \caption{Estimates for single-predicate view (e.g., \rwestate{}).}
    \label{tab:crypto_baselines}
\end{table}

We consider schemes~\cite{kim2016ipe, libert2005ibe} based on prime-order bilinear groups.
We estimate their decryption time by counting the number of bilinear group operations and multiplying by the measured cost of those operations for an efficient bilinear group implementation~\cite[Table 1]{kumar2019jedi}.
For \sys{}, we use results from \figref{fig:selectivity_rwe_state}, dividing the total decryption time for the most populous state by the number of matching rows for that state.
See \tabref{tab:crypto_baselines}.
Our design based on symmetric-key primitives allows \sys{}'s backend to decrypt rows up to three orders of magnitude faster than using off-the-shelf IPE/IBE.
This is separate from \sys{}'s key-hiding tags, which allow skipping rows that cannot be decrypted.

\figref{fig:cloud_revealview} shows that transferring data over the network can dominate the overall time.
\tabref{tab:crypto_baselines} clarifies that this is because \sys{} uses only fast symmetric-key cryptography.
In an IPE- or IBE-based design, decryption time at the client would dominate runtime, particularly for views with high selectivity.

\parhead{Trusted AC Server}
We consider the Trusted AC Server strawman (\secref{sec:background}).
\sys{}'s benefits relative to this strawman are in security---the AC server is an online, central point of attack.
We expect the strawman to outperform \sys{}.

Building an efficient AC Server required forgoing API compatibility with Azure Data Lake Storage.
For example, \texttt{azcopy} scans files ahead of time to determine their lengths before downloading them in chunks; supporting this at the AC service would require downloading the file and computing its view just to know the length of the resulting file, and doing so again to provide the file contents.
Thus, existing tools like \texttt{azcopy} cannot directly interface with our AC server, requiring us to implement a new client to fetch files.

\begin{figure}
    \centering
    \includegraphics[width=\linewidth]{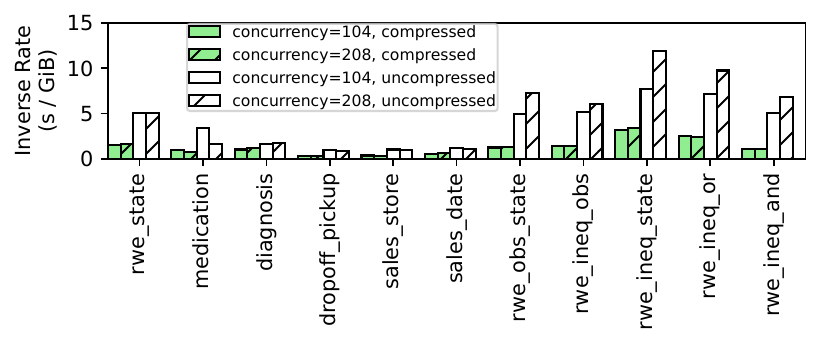}
    \vspace*{-4ex}
    \caption{AC server performance (see also \figref{fig:cloud_revealview}).}
    \label{fig:ac_server}
\end{figure}

We ran the client and AC server on \texttt{E104ids\_v5} instances; results are in \figref{fig:ac_server} (see \figref{fig:cloud_revealview} for comparison).
First, the AC service performs better than \sys{}, as expected (green bars).
But, for uncompressed datasets, its performance is comparable to \sys{} with pipelining; this suggests that the AC server's performance benefits are mostly explained by the fact that it can work with compressed data.
Second, the CPU time at the AC servers was significant, highlighting the need to provision large amounts of compute to deploy an AC server.
Third, 104 concurrent requests (one per logical core) were generally sufficient to achieve the best performance.
Finally, for \sys{} and the AC server, performance is better for more selective views.
For \sys{}, this is due to key-hiding tags; for the AC server, this may be because more selective views are smaller to transfer from AC server to client.

Our AC server uses TLS, but data are not encrypted in cloud storage.
As noted in \secref{sec:model}, one could encrypt cloud storage and have the AC server use a symmetric key to decrypt files on demand.
We expect the performance impact to be small.

\section{Related Work}
\label{sec:related}

As discussed in \secref{sec:intro},
EFSes~\cite{\encfs} enforce access policies based on \emph{public} file paths, at \emph{file}-granularity.
In contrast, \sys{} enforces access policies based on \emph{private}, encrypted, data values, at \emph{cell}-granularity.

As discussed in \secref{sec:intro}, EDBs~\cite{\encdbsys} and ESSes~\cite{\encsearchsys} solve a different problem than \sys{}.
Whereas EDBs (respectively, ESSes) compute the \emph{encrypted} result of a SQL query (respectively, keyword search query) without a decryption key, \sys{} produces the \emph{plaintext} result of \asqlview{} using a decryption key for that view.
Specifically, EDBs and ESSes have two important shortcomings compared to \sys{}.
First, they restrict clients to only issuing certain kinds of queries---those that can be executed cryptographically at the server.
Second, they usually do not support multi-client access control.
We discuss the few exceptions below.

The few EDBs that support access control generally do so based on public, unencrypted attributes~\cite{\cryptacsys}.
For \queryref{query:intro_example} (\secref{sec:intro}), these techniques would require exposing each row's \textsf{State}.
In contrast, \sys{} encrypts all cells and enforces access control based on encrypted cell contents.

CryptDB considers multi-principal access control~\cite[\S{}4]{popa2011cryptdb}.
Compared to \sys{}, it has limitations: (1) it is designed for access policies based on \emph{public data}, requiring principals with access to a row to be listed in that row in \emph{plaintext}, (2) it does not support inequalities or \sqlnsand{}s, and (3) it does not protect users logged in during a compromise.

\textsc{Monomi}~\cite{tu2013monomi} is designed to support any SQL query.
It works by offloading operations that it cannot support cryptographically to the client.
The client, like \sys{}'s compute, is trusted to decrypt the data and see them in plaintext.
Unlike \sys{}, \textsc{Monomi} does not support access control.

Blind Seer~\cite{pappas2014blindseer} supports access control, but only for a single client---a single access policy is applied to \emph{all} queries---and splits the server into multiple trust domains (\textsf{S} and \textsf{IS}).

Curtmola et al.~\cite{curtmola2006improved} propose an ESS, but it is limited to single-word queries and lacks access control.
Protocols based on \textsf{OXT}~\cite{cash2013highlyscalable, jutla2022efficient} or structured encryption~\cite{chase2010structured, kamara2018sql, kamara2020optimal} support more complex queries, but still not access control.
\textsf{MC-OXT}~\cite{jarecki2013ospir} extends \textsf{OXT} to multiple clients, but in a weaker threat model than \sys{}---an adversary who has compromised \emph{both} the storage server \emph{and} some clients can  bypass access control.
\textsf{OXT}-based schemes also use more complex and slower cryptography than \sys{}, whose symmetric-key, high-throughput design is key to data lakes.

A line of work provides cryptographic access control for XML~\cite{bertino-xml-access,suciu-xml,abadi-xml,bertino-survey}.
Unlike \sys{}, these works do not support relational data or data-dependent inequalities.
They may also leak information about documents' tree structure.

Parquet modular encryption~\cite{gershinsky2018efficient, parquetmodularencryption} provides coarse-grained, \emph{column-level} encryption and access control.
In contrast, \sys{} provides data-dependent, \emph{row/cell-level} encryption.

DJoin~\cite[\S{}5.2]{narayan2012djoin} rewrites queries to an intermediate language (IL).
DJoin provides DP, not access control, so its IL has a different structure than \sys{}-canonical form.

Some EDBs~\cite{bajaj2011trusteddb, arasu2015cipherbase, zheng2017opaque, vinayagamurthy2019stealthdb} use a Trusted Execution Environment (TEE) like Intel SGX.
However, TEEs are complex hardware artifacts that are difficult to secure and are vulnerable to side-channel attacks~\cite{bulck2018foreshadow,parno2011memoir,brasser2017cacheattack,schaik2020sgaxe}.
\sys{} moves storage outside of the TCB by placing trust in cryptography.

\section{Discussion and Conclusion}
\label{sec:conclusion}

Data lakes departed from DBMSes by separating compute and storage.
This enables independent scaling of compute and storage, and flexible use of data analysis frameworks (e.g., Spark, Pandas) instead of SQL.

This paper shows that the data lake architecture, originally motivated by scalability and flexibility, actually has positive implications for security.
Specifically, data lakes organize compute and storage into \emph{separate trust domains}.
This allows \sys{} to focus on removing storage entirely from the TCB via fine-grained, data-dependent, cryptographically-enforced access control, without also having to remove compute entirely from the TCB.
Thus, \sys{} can process analytical queries in plaintext, thereby retaining the flexibility of unencrypted data lakes and supporting off-the-shelf frameworks.

This does not, by itself, render EDBs and ESSes insufficient for \sys{}'s purpose.
The actual reason why existing EDBs and ESSes are not viable alternatives to \sys{} is \textbf{access control}; they either do not provide access control at all, or provide weaker access control than \sys{} (\secref{sec:related}).
If an EDB or ESS were to provide access control, it could be adapted to data lakes by applying \sys{}'s system model---specifically, by having data scientists query the EDB/ESS for their \sqlviews{} and then analyze the results in plaintext (\secref{sec:intro}).

That said, \sys{}'s techniques can be applied to existing single-client EDBs/ESSes to enable access control.
This is possible because \sys{}'s goal, namely \emph{cryptographic access control}, is orthogonal to EDBs' and ESSes' goal, namely \emph{query processing over encrypted data}.

As a concrete example, consider CryptDB~\cite{popa2011cryptdb}.
CryptDB includes a multi-user access control design~\cite[\S{}4]{popa2011cryptdb}, in which a user's data is decrypted at the proxy/server when a user logs in, and the proxy/server is trusted to securely delete the users' key and decrypted data when they log out.
As discussed in \secref{sec:related}, CryptDB's design cannot support access policies based on private, encrypted data.
To remedy this, we can apply \sys{}'s techniques, by using \sys{}'s cryptographic protocol as the outermost layer of onion encryption.
Specifically, each cell is encrypted with CryptDB's scheme as usual, and then \sys{} is used to encrypt the resulting ciphertexts, while ensuring that $g(\mathsf{row})$ is evaluated on \emph{plaintext row data}.
When a user logs in, she sends her \sys{} view keys to the proxy/server, which decrypts the matching cells to obtain the CryptDB ciphertexts.
CryptDB can then process queries over these ciphertexts as it is designed to do.
When the user logs out, the CryptDB proxy/server is trusted to delete the users' view keys and the decrypted CryptDB ciphertexts, analogous to CryptDB's current access control.

Similarly to our construction for CryptDB in the previous paragraph, \sys{}'s backend can encrypt data in an ESS for access control.
But \sys{} has even deeper connections to ESSes.
The $\sqlbg{g(\mathsf{row})\sqlin{}?x}$ clauses in \sys{}-canonical form (\secref{sec:planner}) represent a generalized form of \emph{keyword search}, the functionality that ESSes provide.
Rows in \sys{} correspond directly to documents in an ESS, and the expression $g(\mathsf{row})$ corresponds to a generalization of keywords by which documents can be queried.
Thus, \sys{}'s planner, which rewrites complex queries into keyword searches, can actually enable ESSes to process more complex queries.

Given that \sys{}'s backend and planner also apply to EDBs and ESSes, why did we focus on data lakes in this paper?
First, it shows that \sys{}'s techniques are useful independently of the encrypted query processing in EDBs and ESSes and results in a simpler system design for \sys{}.
Second, it shows that \sys{}'s techniques are broadly applicable, as they allow data scientists to use off-the-shelf analytics frameworks and do not restrict the analytical query set as EDBs/ESSes do.
Third, it shows that \sys{}'s design is \emph{synergistic} with current trends in data analytics systems, specifically the separation between compute and storage that has come to define the data lake paradigm.

Finally, \sys{}'s amortized performance overhead can be small, as it only requires decrypting data at the \emph{start} of an interactive session.
Thus, we are hopeful that \sys{}'s techniques can help protect sensitive data used for analytics.

\section*{Acknowledgments}
We would like to thank Joey Gonzalez and Matei Zaharia for formative conversations in the early stages of this research project.
This work is supported by NSF CISE Expeditions
\#CCF-1730628, NSF GRFP \#DGE-1752814, and gifts from Accenture, AMD, Anyscale, Cisco, Google, IBM, Intel, Intesa Sanpaolo, Lambda, Lightspeed, Mibura, Microsoft, NVIDIA, Samsung SDS, and SAP.
Any opinions, findings, and conclusions or recommendations expressed in this material are those of the authors and do not necessarily reflect the views of the National Science Foundation.

\bibliographystyle{plain}
\bibliography{references}

\clearpage
\appendix

\section{Full Backend Protocol Description}\label{app:full_protocol_description}

\subsection{EncryptTable}

\noindent
Input:
\begin{itemize}[leftmargin=*, topsep=0ex, noitemsep]
    \item Partition (index $p$) of table $t$
    \begin{itemize}[leftmargin=*, topsep=0ex, noitemsep]
        \item $m_{r, c}$ denotes value of cell at row $r$ and column $c$
    \end{itemize}
    \item Table key $\tabkey{}$ (chosen randomly; same for all partitions)
\end{itemize}
Outputs: Partition (index $p$) of encrypted table $t'$\\
Algorithm:
Sample $\tabkey{} \stackrel{\$}{\leftarrow} \{0, 1\}^\lambda$.
For $p$ in $1 \ldots n^\mathsf{part}$:
\begin{itemize}[leftmargin=*, topsep=0ex, noitemsep]
    \item For each row index $r$ in $1 \ldots n_{p}^{\mathsf{row}}$:
    \begin{itemize}[leftmargin=*, topsep=0ex, noitemsep]
        \item Derive row key $k_r \definedas{} \prf{\tabkey{}}{p \concat r}$
        \item For each column $c$ in $1 \ldots n^\mathsf{col}$:
        \begin{itemize}[leftmargin=*, topsep=0ex, noitemsep]
            \item Derive cell key $k_{r, c} \definedas{} \prf{k_{r}}{c}$
            \item The value of $t'$ at $r, c$ (in partition $p$) is $\mathsf{OTE}(k_{r, c}, m_{r, c})$
        \end{itemize}
    \end{itemize}
\end{itemize}

\subsection{AddFamily}

\noindent
Inputs:
\begin{itemize}[leftmargin=*, topsep=0ex, noitemsep]
    \item Encrypted partition (index $p$) of table $t$
    \item Table key $\tabkey{}$
    \item Family key $k^{\mathsf{fam}}$ (chosen randomly; same for all partitions)
    \item Indices $c_1, \ldots, c_{n^{\mathsf{proj}}}$ of columns projected from $t$
    \item Predicate functions $g_1 \ldots g_{n^{\mathsf{pred}}}$ used for selection
\end{itemize}
Outputs: Projection/selection/tagging columns for partition\\
Algorithm:
For each row index $r$ in $1 \ldots n_{p}^{\mathsf{row}}$:
\begin{itemize}[leftmargin=*, topsep=0ex, noitemsep]
    \item Compute the value of the projection column.
    \begin{itemize}[leftmargin=*, topsep=0ex, noitemsep]
        \item Define $k_r \definedas{} \prf{\tabkey{}}{p \concat r}$, as in EncryptTable.
        \item Compute $\projkey{r}$:
        \begin{enumerate}[leftmargin=*, topsep=0ex, noitemsep]
            \item If $n^{\mathsf{proj}} = 1$, then  $\projkey{r} = \prf{k_r}{c_1}$
            \item If $n^{\mathsf{proj}} = n^{\mathsf{col}}$, then $\projkey{r} = k_r$
            \item Else, sample $\projkey{r} \stackrel{\$}{\leftarrow} \{0, 1\}^\lambda$
        \end{enumerate}
        \item The value of the projection column at row $r$ is:
        \begin{enumerate}[leftmargin=*, topsep=0ex, noitemsep]
            \item If $n^{\mathsf{proj}} = 1$ or $n^{\mathsf{proj}} = n^{\mathsf{col}}$: $\prf{\projkey{r}}{0}$
            \item Else: $\enc{\projkey{r}}{k_{r, c_1} \concat \ldots \concat k_{r, c_{n^\mathsf{proj}}}} \concat \enc{\projkey{r}}{0}$
        \end{enumerate}
    \end{itemize}
    \item Compute the value of the selection column.
    \begin{itemize}[leftmargin=*, topsep=0ex, noitemsep]
        \item For each predicate $j$ in $1 \ldots n^{\mathsf{pred}}$:
        \begin{itemize}[leftmargin=*, topsep=0ex, noitemsep]
            \item Derive $\predkey{j} \definedas{} \prf{\famkey{}}{j}$
            \item Derive $\selkey{r, j} \definedas{} \prf{\predkey{j}}{g_j(m_{r, 1} \concat \ldots \concat m_{r, n^{\mathsf{col}}})}$
        \end{itemize}
        \item The value of the selection column at row $r$ is $\enc{\prf{\selkey{r, 1}}{0}}{\projkey{r}} \concat \ldots \concat \enc{\prf{\selkey{r, n^{\mathsf{pred}}}}{0}}{\projkey{r}}$
    \end{itemize}
    \item Compute the value of the tagging column.
    \begin{itemize}[leftmargin=*, topsep=0ex, noitemsep]
        \item For each predicate $j$ in $1 \ldots n^{\mathsf{pred}}$:
        \begin{itemize}[leftmargin=*, topsep=0ex, noitemsep]
            \item Derive $\tau_{\selkey{r, j}} \definedas{} \prf{\selkey{r, j}}{p}$.
            \item Derive $\mathsf{tag}_{\selkey{r, j}} \definedas{} \prf{\tau_{\selkey{r, j}}}{\mathsf{count}_{\selkey{r, j}}}$, and truncate it to the desired length.
        \end{itemize}
        \item Tagging column at row $r$ contains $\mathsf{tag}_{\selkey{r, 1}} \concat \ldots \concat \mathsf{tag}_{\selkey{r, n^{\mathsf{pred}}}}$
    \end{itemize}
\end{itemize}

\subsection{ViewGen}

\noindent
Inputs: 
\begin{itemize}[leftmargin=*, topsep=0ex, noitemsep]
    \item View family key $\famkey{}$
    \item List of wildcard values $X^j$ for each predicate $g_j$ in the view
\end{itemize}
Outputs: View key set $\viewkeyset{}$ \\
Algorithm: 
\begin{itemize}[leftmargin=*, topsep=0ex, noitemsep]
    \item Define $\viewkeyset{} \definedas{} []$
    \item For $j$ in $1 \ldots n^{pred}$:
    \begin{itemize}[leftmargin=*, topsep=0ex, noitemsep]
        \item Define $K^\mathsf{pred} \definedas{} []$ and define $\predkey{j} \definedas{} \prf{\famkey{}}{j}$
        \item For $x_i^j$ in $X^j$:
        \begin{itemize}[leftmargin=*, topsep=0ex, noitemsep]
            \item Let $k^{\mathsf{view}}_{j, i} \definedas{} \prf{\predkey{j}}{x_i^j}$ and append $k^{\mathsf{view}}_{j, i}$ to $K^\mathsf{pred}$
        \end{itemize}
        \item Append $K^\mathsf{pred}$ to $\viewkeyset{}$
    \end{itemize}
\end{itemize}

\subsection{RevealView}

\noindent
Inputs:
\begin{itemize}[leftmargin=*, topsep=0ex, noitemsep]
    \item $t^\mathsf{fam}$, the output of AddFamily on table $t$. The partition index is denoted $p$.
    \item View key set $\viewkeyset{}$
\end{itemize}
Outputs: Decrypted view over $t$\\
Algorithm: For each partition $p$ in $t^\mathsf{fam}$:
\begin{itemize}[leftmargin=*, topsep=0ex, noitemsep]
    \item For each key $k^{\mathsf{view}}_{j, i}$, define $\mathsf{count}_{k^{\mathsf{view}}_{j, i}}$ and initialize it to 0.
    \item For each $k^{\mathsf{view}}_{j, i}$ in $\viewkeyset{}$: Define $\tau_{k^{\mathsf{view}}_{j, i}} \definedas{} \prf{k^{\mathsf{view}}_{j, i}}{p}$.
    \item For each key $k^{\mathsf{view}}_{j, i}$, compute its NET as $\prf{\tau_{k^{\mathsf{view}}_{j, i}}}{0}$, truncated to the desired length.
    \item Define a map $N'$ from each NET to a tuple of the corresponding key $k^{\mathsf{view}}_{j, i}$ and the predicate index $j$.
    \item For each row $r$ in $t^\mathsf{fam}$:
    \begin{itemize}[leftmargin=*, topsep=0ex, noitemsep]
        \item For each tag in the tagging column:
        \begin{itemize}[leftmargin=*, topsep=0ex, noitemsep]
            \item Look up the decryption key $k^{\mathsf{view}}_{j, i}$ in $N'$. If not found, continue (i.e., skip to the next tag in the inner loop).
            \item Do $\mathsf{DecryptRow}(t^\mathsf{fam}[r], k^{\mathsf{view}}_{j, i}, j)$. If unsuccessful, continue (i.e., skip to the next tag in the inner loop).
            \item Increment $\mathsf{count}_{k^{\mathsf{view}}_{j, i}}$.
            \item Compute new NET for $k^{\mathsf{view}}_{j, i}$: $\prf{\tau_{k^{\mathsf{view}}_{j, i}}}{\mathsf{count}_{k^{\mathsf{view}}_{j, i}}}$, truncated to the desired length.
            \item Update $N'$ to reflect the new NET (i.e., remove mapping for old NET and add mapping for new NET).
        \end{itemize}
    \end{itemize}
\end{itemize}
\ul{\textsf{DecryptRow}}\\
Inputs:
\begin{itemize}[leftmargin=*, topsep=0ex, noitemsep]
    \item Encrypted row (index $r$) of $t^\mathsf{fam}$
    \item Candidate key $\selkey{}$ and matched predicate index $j$
\end{itemize}
Outputs: Decrypted row of the view $t'$, $t'[r]$\\
Algorithm: 
\begin{itemize}[leftmargin=*, topsep=0ex, noitemsep]
    \item Decrypt the $j$th entry of the selection column to reveal $\projkey{}$: $\projkey{} \definedas{} \dec{\prf{\selkey{}}{0}}{\mathsf{SelCol}[j]}$.
    \item Decrypt projection column with $\projkey{}$ to reveal cell keys:
    \begin{itemize}[leftmargin=*, topsep=0ex, noitemsep]
        \item If the projection column has one element, check that $\prf{\projkey{}}{0} = \mathsf{ProjCol}[1]$. Otherwise, check that $\dec{\projkey{}}{\mathsf{ProjCol}[2]} = 0$. If the check fails, abort.
        \item If the view includes only one column $c_1$ (so $n^\mathsf{proj} = 1$), then
        decrypt the value of $t'[r]$ as $\dec{\projkey{}}{t^\mathsf{fam}[r][c_1]}$.
        \item Else, if the view includes all columns (so $n^\mathsf{proj} = n^\mathsf{col}$), then
        re-derive the cell keys from $\projkey{}$. For each non-family column $c$ included in the view:
        \begin{itemize}[leftmargin=*, topsep=0ex, noitemsep]
            \item Derive the cell key $k_{r, c} \definedas{} \prf{\projkey{}}{c}$.
            \item Decrypt the value of $t'[r][c]$ as $\dec{k_{r, c}}{t^\mathsf{fam}[r][c]}$.
        \end{itemize}
        \item Else, decrypt the cell keys $k_{r, c_1} \concat \ldots \concat k_{r, c_{n^\mathsf{proj}}} = \dec{\projkey{}}{\mathsf{ProjCol}[2]}$. For each column $c$ in the view:
        \begin{itemize}[leftmargin=*, topsep=0ex, noitemsep]
            \item Decrypt the value of $t'[r][c]$ as $\dec{k_{r, c}}{t^\mathsf{fam}[r][c]}$.
        \end{itemize}
    \end{itemize}
\end{itemize}

\section{Cryptographic Treatment of \sys{}}
\label{app:formalism}

We now provide a cryptographic treatment of \sys{} and its security guarantees (described informally in \secref{s:informal_security}).

\subsection{System Model}

We define an EDL scheme as follows.

\begin{definition}
\label{def:edl}
An EDL (``encrypted data lake'') scheme is a tuple of four algorithms:
\begin{itemize}[leftmargin=*,topsep=0ex,noitemsep]
    \item $\mathsf{EncryptTable}(t, \tabkey{}) \rightarrow t'$
    \item $\mathsf{AddFamily}(t, \tabkey{}, \family{}, \famkey{}) \rightarrow t'$
    \item $\mathsf{ViewGen}(\view{}, \famkey{}) \rightarrow \viewkey{}$
    \item $\mathsf{RevealView}(t, \viewkey{}) \rightarrow t'$
\end{itemize}
\end{definition}

The syntax matches \figref{fig:api}, except that $\tabkey{}$ and $\famkey{}$ are \emph{inputs}, not \emph{outputs} (as in \secref{sec:backend}).
With the above syntax, the caller is assumed to sample $\tabkey{}$ uniformly at random before calling \textsf{EncryptTable}, and to sample $\famkey{}$ uniformly at random before calling \textsf{AddFamily}.
This is analogous to the caller sampling a symmetric key uniformly at random before invoking symmetric-key encryption.

Using this syntax allows our cryptographic formalism to model the case where \textsf{ViewGen} is called for a view family before \textsf{EncryptTable} or \textsf{AddFamily} are called (i.e., where a view key is generated for a view family before that view family is instantiated in a table).
In that sense, this syntax is more general than the one given in \figref{fig:api}.
\figref{fig:api} presents the API as it does because it is more suggestive of \sys{}'s intended use case, and is more similar to the API actually provided by our implementation.

We require the intuitive notion of ``completeness'' that \textsf{RevealView} indeed produces the same result as materializing the view in plaintext.
We define completeness as follows.

\begin{definition}
\label{def:edl_completeness}
An EDL scheme is said to be complete if for any table $t$, \viewfamily{} $f$, and \sqlview{} $v$ where $v \in f$, the following holds:\\

\textbf{If} we run $\mathsf{EncryptTable}(t, k^t) \rightarrow t'$ \textbf{and} $\mathsf{AddFamily}(t', k^t, f, k^f) \rightarrow t''$ (and possibly other \textsf{AddFamily} operations on the table) \textbf{and} $\mathsf{ViewGen}(v, k^f) \rightarrow k^v$ \textbf{then} $\mathsf{RevealView}(t'', k^v) = v(t)$ \\

where $v(t)$ denotes the view $v$ applied to the table $t$, and where $k^t$ and $k^f$ are sampled uniformly at random.
\end{definition}

\subsection{Security Definition}

In defining security, we consider \viewfamilies{} with the constraint that \aviewfamily{} must $\sqlnsselect$ all columns referenced in its $\sqlnswhere$ clause.
This means that the set of cell positions described by \asqlview{} includes both the cell positions that that the view reveals and the cell positions that the $\sqlnswhere$ clause references for rows where cells are revealed (as discussed in \secref{s:supported_sql}).

Our main security definition is for a notion that we call \emph{selective security}.

\begin{definition}[Selective Security EDL Game]
\label{def:selective_security_game}
Selective security for an EDL scheme is defined in terms of the following game between an adversary $\adv{}$ and challenger $\chl{}$.

\noindent
\textbf{Initialization.}
$\adv{}$ chooses the schemas of $u$ relations and the size $n_i$ for each relation (for $1 \leq i \leq u$), where $n_i$ includes the number of partitions and size of each partition.
$\adv{}$ also chooses a set $P$ of cell positions.
Each cell position is a tuple $(i, p, r, c)$ where $i$ is the index of a relation, $p$ is a partition ID, $r$ is the row index, and $c$ is the column index.
$\adv{}$ ``declares'' $u$, its chosen schemas, $n_i$ ($\forall 1 \leq i \leq u$), and $P$, by sending them to $\chl{}$.\\
\textbf{Phase 1.}
$\adv{}$ repeatedly issues queries to $\chl{}$.
There are two types of queries that $\adv{}$ may make:
\begin{itemize}[leftmargin=*,noitemsep,topsep=0ex]
    \item $\adv{}$ asks $\chl{}$ to instantiate a view family in one of the relations.
    $\chl{}$ samples the corresponding $\famkey{}$ uniformly at random but does not give $\famkey{}$ to $\adv{}$.
    \item $\adv{}$ asks $\chl{}$ to generate a key for a view $v$ belonging to a view family specified in a previous query (identified by the ID of the earlier query).
    $\chl{}$ generates the view key by calling \textsf{ViewGen} on the $\famkey{}$ generated in the earlier query and gives the resulting view key to $\adv{}$.
\end{itemize}
\textbf{Challenge.} $\adv{}$ chooses two $u$-length tuples of relations, \reltup{0} and \reltup{1}.
The choice is subject to the following restrictions:
\begin{enumerate}[leftmargin=*,noitemsep,topsep=0ex]
    \item \label{constraint:sel_attack_set} A cell in \reltup{0} and a cell in \reltup{1}, both at position $(i, p, r, c)$, must contain identical data if $(i, p, r, c) \notin P$, and must contain data of the same length if $(i, p, r, c) \in P$.
    \item \label{constraint:sel_views} Each requested view $v$ must describe the same set of cell positions, denoted $v^\mathcal{R}$, whether applied to \reltup{0} or \reltup{1}, and those cell positions must be disjoint from $P$ (i.e., it must hold that $v^\mathcal{R} \cap P = \varnothing$).
\end{enumerate}
$\adv{}$ sends \reltup{0} and \reltup{1} to $\chl{}$.
$\chl{}$ chooses a random bit $b\stackrel{\$}{\leftarrow}\{0, 1\}$.
Then, it encrypts \reltup{b} according to \sys{}'s protocol, by calling $\textsf{EncryptTable}$ on each table (with a randomly sampled $\tabkey{}$ for each table).
For each view family specified in the queries in Phase 1, it calls $\textsf{AddFamily}$ to instantiate the view family in the encrypted \reltup{b}, using the $\tabkey{}$ for the specified table and the $\famkey{}$ for that view family.
Finally, it sends the resulting encrypted relations, $t_1, \ldots, t_u$, to $\adv{}$.\\
\textbf{Phase 2.} $\adv{}$ can issue additional queries of the same form as those in Phase 1, subject to Constraint \#\ref{constraint:sel_views} in the Challenge phase.
When $\adv{}$ requests a view family, $\chl{}$ instantiates the view family by calling $\textsf{AddFamily}$ and responds to $\adv{}$ with the updated table $t'$ right away.
As in Phase 1, $\adv{}$ does not reveal $\famkey{}$ to $\chl{}$.\\
\textbf{Guess.} $\adv{}$ outputs $b' \in \{0, 1\}$, and wins the game if $b = b'$.
The advantage of the adversary $\adv{}$ is defined as $\left|\Pr[\adv{}\text{ wins}] - \frac{1}{2}\right|$.
\end{definition}

Now, we define selective security in terms of that game.

\begin{definition}
\label{def:selective_security}
An EDL scheme is \emph{selectively secure} if, for any non-uniform probabilistic polynomial-time adversary $\adv$, it holds that $\adv$'s advantage in the Selective Security EDL Game (Definition~\ref{def:selective_security_game}) is negligible.
\end{definition}

\subsubsection{Discussion of Our Security Definition}

We refer to our notion of security as \emph{selective security} because we require $\adv$ to \emph{select}, ahead of time, which cells to attack.
The cells that the adversary chooses to attack are those in $P$, and the adversary is not allowed to query cells that they are choosing to attack (i.e., any set $P_v$ of queried cells positions must be disjoint from $P$).
This is analogous to \emph{selective} security definitions used in the context of Identity-Based Encryption~\cite{boneh2005hibe} and Attribute-Based Encryption~\cite{goyal2006kpabe}, where the adversary must select, ahead of time, which ID or set of attributes to attack, and may not query the private key corresponding to that ID or those attributes.

In a fully \emph{adaptive} notion of security, the adversary would not have to declare the set $P$ of cell positions up front, and would not be restricted by $P$ in the Queries phase.
Because our selective security definition requires the adversary to declare $P$ during the Initialization phase, with the analogous restrictions in the Queries phase, it is weaker (i.e., protects against a weaker adversary) than an adaptive security definition would be.
We use a selective notion of security rather than a fully adaptive one because adaptive notions of security are difficult to achieve in practice.

\subsubsection{Comparison to Static Security}

A commonly used weaker alternative to adaptive security is static security.
In static security definitions, the adversary $\adv$ commits up front to a sequence of queries that she will make to $\chl{}$.
While our selective security definition may be weaker than fully adaptive security, our selective security definition is at least as strong as a static security definition.
The intuition for this is that, given a sequence of queries declared up front by a static adversary, one can compute the set $P$ of cell positions to declare up front in the Selective Security EDL Game.

To present this argument more formally, we first provide a formal definition of static security.
This allows us to prove via a reduction that selective security is at least as strong as static security.

\begin{definition}[Static Security EDL Game]
\label{def:static_security_game}
Static security for an EDL scheme is defined in terms of the following game.

\noindent
\textbf{Initialization.}
$\adv{}$ chooses the schemas of $u$ relations and the size $n_i$ for each relation (for $1 \leq i \leq u$), where $n_i$ includes the number of partitions and size of each partition.
$\adv{}$ chooses $F$ and $V$, defined as follows:
\begin{itemize}[leftmargin=*,noitemsep,topsep=0ex]
    \item $F$ is a list of view families, including tables (specified by index) in which $\chl{}$ must instantiate each of them.
    \item $V$ maps each view family in $F$ to a list of views in that view family, whose keys $\chl{}$ must generate and reveal to $\adv{}$.
\end{itemize}
Each cell position is a tuple $(i, p, r, c)$ where $i$ is the index of a relation, $p$ is a partition ID, $r$ is the row index, and $c$ is the column index.
$\adv{}$ ``declares'' $u$, its chosen schemas, $n_i$ ($\forall 1 \leq i \leq u$), $F$, and $V$, by sending them to $\chl{}$.\\
\textbf{Challenge.} $\adv{}$ chooses two $u$-length tuples of relations, \reltup{0} and \reltup{1}, subject to the following restrictions:
\begin{enumerate}[leftmargin=*,noitemsep,topsep=0ex]
    \item \label{constraint:static_attack_set} A cell in \reltup{0} and a cell in \reltup{1}, both at position $(i, p, r, c)$, must contain identical data if any view $v \in V$ describes $(i, p, r, c)$, and must contain data of the same length if no view in $V$ describes $(i, p, r, c)$.
    \item \label{constraint:static_views} Each view $v \in V$ must describe the same set of cell positions, denoted $v^\mathcal{R}$, whether applied to \reltup{0} or \reltup{1}.
\end{enumerate}
$\adv{}$ sends \reltup{0} and \reltup{1} to $\chl{}$.
$\chl{}$ chooses a random bit $b\stackrel{\$}{\leftarrow}\{0, 1\}$.
Then, it encrypts \reltup{b} according to \sys{}'s protocol, by calling $\textsf{EncryptTable}$ on each table (with a randomly sampled $\tabkey{}$ for each table).
For each view family in $F$, it invokes \textsf{AddFamily} to instantiate the view family in the encrypted \reltup{b}, using the $\tabkey{}$ for the specified table and a randomly sampled $\famkey{}$ for each view family.
Finally, it sends the resulting encrypted relations, $t_1, \ldots, t_u$, to $\adv{}$.
$\chl{}$ does not send the family keys $\famkey{}$ for the view families in $F$ to $\adv{}$.
$\chl{}$ invokes \textsf{ViewGen} for each view in $V$ (using the corresponding family key $\famkey{}$) and sends the resulting view keys to $\adv{}$.\\
\textbf{Guess.} $\adv{}$ outputs $b' \in \{0, 1\}$, and wins the game if $b = b'$.
The advantage of the adversary $\adv{}$ is defined as $\left|\Pr[\adv{}\text{ wins}] - \frac{1}{2}\right|$.
\end{definition}

Now, we define static security in terms of that game.

\begin{definition}
\label{def:static_security}
An EDL scheme is \emph{statically secure} if, for any non-uniform probabilistic polynomial-time adversary $\adv$, it holds that $\adv$'s advantage in the Static Security EDL Game (Definition~\ref{def:static_security_game}) is negligible.
\end{definition}

Using these definitions, we state and prove, in formal terms, that selective security is at least as strong as static security.

\begin{theorem}
If an EDL scheme is selectively secure (Definition~\ref{def:selective_security}), then it is statically secure (Definition~\ref{def:static_security}).
\end{theorem}
\begin{proof}
We prove this statement using contraposition---we show that, if a non-uniform probabilistic polynomial-time adversary $\advstatic{}$ can win the Static Security EDL Game (Definition~\ref{def:static_security_game}) with non-negligible probability for the EDL scheme, then there exists a non-uniform probabilistic polynomial-time adversary $\advsel{}$ that can win the Selective Security EDL Game (Definition~\ref{def:selective_security_game}) with non-negligible probability for the EDL scheme.
We do so using a reduction, constructing $\advsel{}$ as an algorithm that uses $\advstatic{}$ as a black box.
The following paragraphs describe $\advsel{}$.

For the \textbf{Initialization} phase, $\advsel{}$ observes the output of $\advstatic{}$ in its Initialization and Challenge phases and ``declares'' the same values for $u$, the relations' schemas, and $n_i$ ($\forall 1 \leq i \leq u$), by sending them to $\chl{}$.
$\advsel{}$ also ``declares'' the set $P$, computed as follows.
Because \reltup{0} and \reltup{1} have the same number of relations, and corresponding relations have the same size, partitioning, and schema, the set of all possible cell positions is the same for both; let $T$ refer to this set of all possible cell positions.
For each view $v \in V$, $\advsel{}$ computes $v^{\mathcal{R}}$, the set of cell positions revealed by that view (which both games require to be the same in both \reltup{0} and \reltup{1}).
Then $\advsel{}$ computes $Q$, a set of cell positions, as
$$Q=\bigcup_{v \in V} v^{\mathcal{R}}$$
and computes $P$ as $P=T\setminus{}Q$.
In effect, $P$ is the set of cell positions \emph{not} covered by any of the views chosen by $\advstatic{}$.

In \textbf{Phase 1}, $\advsel{}$ issues queries corresponding to the values $F$ and $V$ output by $\advstatic{}$ during its Initialization phase.

In the \textbf{Challenge} phase, $\advsel{}$ sends $\chl{}$ the same values for \reltup{0} and \reltup{1} as $\advstatic{}$ would send to its challenger.
If $\advstatic{}$ outputs valid \reltup{0} and \reltup{1} that satisfy the requirements of the Challenge phase in the Static Security EDL Game, then the same \reltup{0} and \reltup{1}, output by $\advsel{}$, will satisfy the requirements of the Challenge phase in the Selective Security EDL Game:
\begin{enumerate}[leftmargin=*,noitemsep,topsep=0ex]
    \item $P$ is chosen such that the cell positions in $P$ are exactly those that are not described by any view $v \in V$. Therefore, ``any view $v \in V$ describes $(i, p, r, c)$'' is equivalent to ``$(i, p, r, c) \notin P$''. Thus, if \reltup{0} and \reltup{1} satisfy Constraint \#\ref{constraint:static_attack_set} in the Static Security EDL Game, then they satisfy Constraint \#\ref{constraint:sel_attack_set} in the Selective Security EDL Game.
    \item Constraint \#\ref{constraint:sel_views} of the Selective Security EDL Game has two conditions.
    The first condition is identical to Constraint \#\ref{constraint:static_views} in the Static Security EDL Game.
    The second condition holds because of how we constructed $P$---by definition, $Q \not\subseteq P$, and $\forall v \in V$ $v^\mathcal{R} \subseteq Q$, so $\forall v \in V$ $v^\mathcal{R} \cap P = \varnothing$.
    Therefore, if \reltup{0} and \reltup{1} satisfy Constraint \#\ref{constraint:static_views} in the Static Security EDL Game, then they satisfy Constraint \#\ref{constraint:sel_views} in the Selective Security EDL Game.
\end{enumerate}
$\advsel{}$ receives the encrypted relations from $\chl{}$ and sends them to $\advstatic{}$.
$\advsel{}$ also gives $\advstatic{}$ the view keys it obtained from $\chl{}$ in Phase 1.

In \textbf{Phase 2}, $\advsel{}$ does not issue any queries.

As its \textbf{Guess}, $\advsel{}$ outputs the same bit $b'$ as $\advstatic{}$.
The information that $\advsel{}$ gave $\advstatic{}$ is distributed identically to what $\advstatic{}$ would receive from a challenger in the Static Security EDL Game who chose the same value of $b$ as $\chl{}$ did.
Therefore, $\advsel{}$ has the same advantage in the Selective Security EDL Game as $\advstatic{}$ does in the Static Security EDL Game.
\end{proof}

\subsection{Preliminaries}
\label{s:formalism_preliminaries}

\sys{}'s protocol depends on three cryptographic primitives: a pseudorandom function (PRF), an encryption scheme $\encsymb{}$, and a one-time encryption scheme $\onetimeencsymb{}$.
Here, we discuss these cryptographic primitives, their security guarantees, and how we instantiate them in our \sys{} implementation.

\subsubsection{Pseudorandom Functions}

A \emph{pseudorandom function} (PRF) is a deterministic function that accepts as input a key $k$ and a message $x$; we denote its application as $\prf{k}{x}$.
A formal treatment of PRFs is given by Boneh and Shoup~\cite[Definition 4.2]{boneh2023blockciphers}.

In our \sys{} implementation, we instantiate PRFs in two different ways.

The first way is to simply use the AES block cipher as a PRF.
This is a valid approach because the Switching Lemma~\cite[Theorem 4.4]{boneh2023blockciphers} guarantees that if AES is a secure block cipher (pseudorandom permutation), then it is also a secure PRF.
This technique is very efficient; in particular, when using AES-128 as a PRF, the key size is the same as the block size, allowing the PRF's output to be used directly as the key to another PRF invocation.
Unfortunately, this approach limits the size of the PRF input to the block size (e.g., 16 bytes in the case of AES-128), so we cannot instantiate every PRF in \sys{}'s protocol in this way.

The second way, which supports arbitrary-size inputs, is to use CBC-MAC, instantiated with the AES block cipher, together with prepending the input length to the input.
This approach is valid because the CBC-MAC construction, alone, produces a prefix-free PRF~\cite[Theorem 6.3]{boneh2023integrity}, so when it is used with a prefix-free encoding of the input, as is obtained when prepending the input's length, it produces a fully secure PRF~\cite[Theorem 6.8]{boneh2023integrity}.

\subsubsection{Symmetric-Key Encryption}
\sys{} requires a symmetric-key encryption scheme $\encsymb$ with two properties: CPA-security and key-privacy.
A formal treatment of CPA-security is given by Boneh and Shoup~\cite[Theorem 5.2]{boneh2023cpa}.
Bellare et al.~\cite{bellare2001keyprivacy} coin the term \emph{key-privacy} in the public-key setting; here we require an analogous property in the symmetric-key setting.
Abadi and Rogaway~\cite[Definition 2]{abadi2000twoviews} provide a formal security definition for a symmetric-key encryption scheme that is both CPA-secure and key-private; they refer to key-private encryption as ``which-key concealing'' and to such CPA-secure and key-private symmetric-key encryption schemes as ``type-1 secure.''

In our \sys{} implementation, we instantiate the symmetric-key encryption scheme $\encsymb{}$ by using the AES block cipher in CTR mode.
Abadi and Rogaway~\cite[Section 4.4]{abadi2000twoviews} explain that CTR-mode encryption is indeed ``type-1 secure.''

\subsubsection{One-Time Symmetric-Key Encryption}
\sys{} makes use of a one-time symmetric-key encryption scheme $\onetimeencsymb$.
It does so for efficiency; it would be correct and secure to instantiate $\onetimeencsymb$ in exactly the same way as we instantiated $\encsymb$, but the idea is that $\onetimeencsymb$ can be instantiated in a more efficient way.
This is possible because, unlike $\encsymb$, $\onetimeencsymb$ need not support key reuse.
Specifically, $\onetimeencsymb$ is semantically secure, as defined by Boneh and Shoup~\cite[Definition 2.2]{boneh2023symmetrickey}, but not necessarily CPA-secure.

In our \sys{} implementation, we instantiate $\onetimeencsymb$ by using AES in CTR mode for messages longer than the key (just as in $\encsymb{}$), but using the one-time pad scheme~\cite[Example 2.2]{boneh2023symmetrickey}, which is more efficient, for short messages.

\subsection{Security Guarantee and Proof of Security}
Now, we state and sketch a proof of a theorem describing \sys{}'s security.

\begin{theorem}[\sys{}'s Security Guarantee]
If \sys{} is instantiated with a secure PRF, a CPA-secure and key-private encryption scheme $\encsymb{}$, and a semantically secure one-time encryption scheme $\onetimeencsymb{}$, then \sys{} is a selectively secure EDL scheme under Definition \ref{def:selective_security}.
\end{theorem}

\begin{proof}[Proof Sketch]
We make a hybrid argument, presenting a sequence of hybrid games that present the same interface to $\adv{}$ as the Selective Security EDL Game but in which $\chl{}$ replies with differently formed messages.
The first hybrid $\hyb{0}$ is identical to the Selective Security EDL Game, and the final hybrid $\hyb{*}$ is one where $\adv{}$'s advantage is $0$ by construction.
We argue that for any two adjacent hybrid games in the sequence $\hyb{j}$ and $\hyb{j + 1}$, the difference in $\adv{}$'s advantage is negligible.
Because the number of hybrid games is polynomial in the security parameter $\secp{}$, this implies that $\adv{}$'s advantage in the game $\hyb{0}$ is negligible, as desired.

We now present the sequence of hybrid games.
In each step except the last, we change only how $\chl{}$ responds to queries.
In some cases, these changes are localized to only how certain rows are processed in an \textsf{AddFamily} operation for certain \viewfamilies{}.
We refer to such combinations of rows and \viewfamilies{} as \textbf{\textit{critical combinations}}.
Specifically, \aviewfamily{} $f$ and a row $(p ,r)$ form a critical combination if $f$ $\sqlnsselect{}$s some column $(i, c)$ such that $(i, p, r, c) \in P$.
Our hybrid games in this proof sketch should be interpreted as key stages; between each pair of stages are multiple hybrid games, where only one instance of a cryptographic primitive is changed at a time.

\parhead{Hybrid $\hyb{0}$}
This game is exactly the Selective Security EDL Game (i.e., Definition~\ref{def:selective_security_game}).

\parhead{Hybrid $\hyb{1}$}
This is the same as $\hyb{0}$ except that we replace every row key with a truly random value.
The difference in $\adv$'s advantage between the previous stage and this stage is negligible because of the security of a PRF keyed by the truly random table key, together with the fact that each row key for a table is generated using a different input to the PRF (namely $r \concat c$).

\parhead{Hybrid $\hyb{2}$}
This is the same as $\hyb{1}$ except that we replace each predicate key $\predkey{j}$ with a truly random value.
The difference in $\adv$'s advantage between the previous stage and this stage is negligible because of the security of a PRF keyed by the truly random view family key $\famkey{}$, together with the fact that each predicate key $\predkey{j}$ for a given \viewfamily{} is generated using a different input to the PRF (namely $j$).

\parhead{Hybrid $\hyb{3}$}
This is the same as $\hyb{2}$ except that we replace PRFs keyed on each predicate key $\predkey{j}$ with truly random functions.
The difference in $\adv$'s advantage between the previous stage and this stage is negligible because of the security of a PRF keyed by the truly random predicate keys.
Note that, although the PRFs keyed on $\predkey{j}$ are now replaced with uniformly random functions and $\chl$ uses these random function to derive the selection keys $\selkey{r, j}$, this is \emph{not} equivalent to $\chl$ sampling each row's selection key uniformly at random.
Specifically, selection keys may still repeat across rows; this is because multiple rows may have the same value of $g_j(\mathsf{row})$, which is used as the input to the PRF.

\parhead{Hybrid $\hyb{4}$}
This is the same as $\hyb{3}$ except that, for critical combinations of \viewfamilies{} and rows, we replace PRFs keyed on selection keys $\selkey{r, j}$ with truly random functions when $\chl$ runs \textsf{AddFamily}.
Specifically, we associate each selection key with its own truly random function, and replace each PRF invocation keyed on that selection key with an invocation of its truly random function.
This impacts the derivation of encryption keys for the selection layer (used to encrypt the projection keys) and the derivation of each partition's tagging key, for those critical combinations.
The difference in $\adv$'s advantage between the previous stage and this stage is negligible because of the security of PRFs keyed by the selection keys, which are sampled randomly (due to $\hyb{3}$).
This requires that, for these critical combinations, the selection keys are not revealed to $\adv$; this holds because $\adv$ may only request view keys for which $P_v \cap P = \varnothing$, which implies that none of $\adv$'s requested view keys includes a $\viewkey_{j, x}$ matching a selection key for a critical combination as part of $\adv$'s requested view keys.
In order for this to hold, it is important that views must $\sqlnsselect{}$ columns referenced in their $\sqlnswhere{}$ clause (i.e., $P_v$ includes the cell positions that the $\sqlnswhere{}$ clause references for rows where cells are revealed), as mentioned in \secref{s:supported_sql}.

\parhead{Hybrid $\hyb{5}$}
This is the same as $\hyb{4}$ except that, for critical combinations of \viewfamilies{} and rows, we replace tags with random strings when running $\textsf{AddFamily}$.
The difference in $\adv$'s advantage between the previous stage and this stage is negligible because of the security of a PRF keyed on the tagging keys, which are the result of a random function due to $\hyb{4}$; while the tagging keys for a predicate may repeat among rows, the combination of tagging key and counter is always different when generating each tag (i.e., the counters guarantee that the PRF invocation to generate the tag is always performed with a different counter when the key is reused).

\parhead{Hybrid $\hyb{6}$}
This is the same as $\hyb{5}$ except that, for critical combinations of \viewfamilies{} and rows, we replace each encryption of the projection key in the selection column with an encryption of a ``zero string'' under the same key, when $\chl$ runs $\textsf{AddFamily}$.
The difference in $\adv$'s advantage between the previous stage and this stage is negligible because of the CPA-security of the encryption scheme $\encsymb{}$.

\parhead{Hybrid $\hyb{7}$}
This is the same as $\hyb{6}$ except that, for critical combinations of AC view families and rows, we replace encryptions of ``zero strings'' in the selection column with encryptions of ``zero strings'' \emph{under random keys}.
The difference in $\adv$'s advantage between the previous stage and this stage is negligible because of the key-privacy of the encryption scheme $\encsymb{}$.
Observe that, at this hybrid stage, the ciphertexts in the selection column, for critical combinations, are entirely independent of the data chosen by $\adv$ at cell positions in $P$.

\parhead{Hybrid $\hyb{8}$}
This is the same as $\hyb{7}$ except that, for critical combinations of \viewfamilies{} and rows, we replace the encryption of cell keys with an encryption of ``zero strings'' of the same length.
In some optimized cases, the encryption of cell keys is not present; for such view families, this hybrid step changes nothing compared to $\hyb{6}$.
The difference in $\adv$'s advantage between the previous stage and this stage is negligible because of the CPA security of the encryption scheme when using $\projkey{r}$ as the encryption key (since, for critical combinations, the encryption of $\projkey{r}$ has been replaced with an encryption of a ``zero string,'' meaning that $\projkey{r}$ is never revealed to $\adv$).

\parhead{Hybrid $\hyb{9}$}
This is the same as $\hyb{8}$ except that, for cell positions in $P$, we replace every cell key with a truly random value.
The difference in $\adv$'s advantage between the previous stage and this stage is negligible because of the security of a PRF keyed by the truly random row keys (after the $\hyb{1}$ step).
Importantly, after the $\hyb{7}$ step, row keys are no longer revealed to $\adv$ for critical combinations, even in the optimized case where $\projkey{r} = k_r$.
Note that, for this to work in the optimized case where the projection key is the row key, replacing the encryption of 0 with a PRF invocation at 0 is crucial.

\parhead{Hybrid $\hyb{10}$}
This is the same as $\hyb{9}$ except that cells at positions in $P$, which normally contain $\onetimeencsymb{}$ encryptions of cell data, are changed to instead contain $\onetimeencsymb{}$ encryptions of ``zero strings'' of the same length.
The difference in $\adv$'s advantage between the previous stage and this stage is negligible because of the semantic security of the one-time encryption scheme $\onetimeencsymb{}$, together with the fact that the cell keys are random (due to $\hyb{9}$) and never reused.

In $\hyb{10}$, the data chosen by $\adv$ at cell positions in $P$ have \emph{no influence} on the values that $\chl$ gives to $\adv$ in response to any query; this is because the encryptions of those data have been replaced with encryptions of zero strings, and any keys that are derived from them have been replaced with random values.
The distribution of data that $\chl$ gives to $\adv$ is identical whether $b = 0$ or $b = 1$; in particular, the data chosen by $\adv$ at cell positions outside of $P$ are identical in \reltup{0} and \reltup{1}.
Therefore, $\adv{}$'s advantage in the $\hyb{10}$ game is $0$.
We take $\hyb{*} = \hyb{10}$, completing the proof sketch.
\end{proof}

\subsection{Discussion}
The above covers the security of \sys{}'s cryptographic backend protocol.
Our implementation of \sys{} uses a collision-resistant hash to support inequality operations on strings.
It does so by hashing strings to integers, to leverage \sys{}'s support for inequalities over integer quantities.
This use of collision-resistant hashing is not covered by our formal definition above because collision-resistant hashing is merely used as a wrapper around \sys{}'s cryptographic protocol; \sys{}'s cryptographic backend protocol does not itself use collision-resistant hashing.
In particular, any mechanism to map each string in a column to a unique integer would be sufficient for use with \sys{}.

\section{Analytical Queries}
\label{app:analytical_queries}

\noindent
\begin{minipage}{\linewidth}
\begin{lstlisting}[language=sql, breaklines=true, caption={Analytical queries over the Synthea dataset.}, label={lst:synthea-analytical}, captionpos=b]
--Based on Section 9.3 of 
--https://ohdsi.github.io/TheBookOfOhdsi/SqlAndR.html.
--Query 1
SELECT AVG(DATEDIFF(DAY,
                    observation_period_start_date,
                    observation_period_end_date) / 365.25) 
                    AS num_years
FROM (SELECT
    MIN(ENCOUNTER_START) AS observation_period_start_date,
    MAX(ENCOUNTER_STOP) AS observation_period_end_date
    FROM rwe_state
    GROUP BY OBSERVATION_PATIENT);

--Query 2
SELECT COUNT(DISTINCT OBSERVATION_PATIENT) FROM rwe_state;

--Query 3
SELECT MAX(YEAR(observation_period_end_date) 
        - YEAR(date_of_birth)) AS max_age
FROM (SELECT
    MAX(ENCOUNTER_STOP) AS observation_period_end_date,
    first(PATIENT_BIRTHDATE) AS date_of_birth
    FROM rwe_state
    GROUP BY OBSERVATION_PATIENT);

--Query 4
WITH ages
AS (
    SELECT age,
        ROW_NUMBER() OVER (
            ORDER BY age
            ) order_nr
    FROM (
        SELECT YEAR(observation_period_end_date) 
            - YEAR(date_of_birth) AS age
        FROM (SELECT
            MAX(ENCOUNTER_STOP) 
                AS observation_period_end_date,
            first(PATIENT_BIRTHDATE) AS date_of_birth
            FROM rwe_state
            GROUP BY OBSERVATION_PATIENT)
        ) age_computed
    )
SELECT MIN(age) AS min_age,
    MIN(CASE
            WHEN order_nr < .25 * n
                THEN 9999
            ELSE age
            END) AS q25_age,
    MIN(CASE
            WHEN order_nr < .50 * n
                THEN 9999
            ELSE age
            END) AS median_age,
    MIN(CASE
            WHEN order_nr < .75 * n
                THEN 9999
            ELSE age
            END) AS q75_age,
    MAX(age) AS max_age
FROM ages
CROSS JOIN (
    SELECT COUNT(*) AS n
    FROM ages
    ) population_size;
\end{lstlisting}
\end{minipage}

\noindent
\begin{minipage}{\linewidth}
\begin{lstlisting}[language=sql, breaklines=true, caption={Analytical queries over the LHBench TPC-DS dataset.}, label={lst:lhbench-analytical}, captionpos=b]
--Query 5 (TPC-DS Query 3)
SELECT dt.d_year, 
               item.i_brand_id          brand_id, 
               item.i_brand             brand, 
               Sum(ss_ext_discount_amt) sum_agg 
FROM   date_dim dt, 
       sales_store store_sales, 
       item 
WHERE  dt.d_date_sk = store_sales.ss_sold_date_sk 
       AND store_sales.ss_item_sk = item.i_item_sk 
       AND item.i_manufact_id = 427 
       AND dt.d_moy = 11 
GROUP  BY dt.d_year, 
          item.i_brand, 
          item.i_brand_id 
ORDER  BY dt.d_year, 
          sum_agg DESC, 
          brand_id
LIMIT 100;

--Query 6 (TPC-DS Query 7)
SELECT i_item_id,
      Avg(ss_quantity) agg1,
      Avg(ss_list_price) agg2,
      Avg(ss_coupon_amt) agg3,
      Avg(ss_sales_price) agg4
FROM  sales_store,
      customer_demographics,
      item, date_dim, promotion
WHERE  ss_sold_date_sk = d_date_sk
      AND ss_item_sk = i_item_sk
      AND ss_cdemo_sk = cd_demo_sk
      AND ss_promo_sk = p_promo_sk
      AND cd_gender = 'F'
      AND cd_marital_status = 'W'
      AND cd_education_status = '2 yr Degree'
      AND ( p_channel_email = 'N'
         	OR p_channel_event = 'N' )
      AND d_year = 1998
GROUP BY i_item_id ORDER BY i_item_id
LIMIT 100;

--Query 7 (TPC-DS Query 6)
SELECT a.ca_state state, 
               Count(*)   cnt 
FROM   customer_address a, 
       customer c, 
       sales_store s, 
       date_dim d, 
       item i 
WHERE  a.ca_address_sk = c.c_current_addr_sk 
       AND c.c_customer_sk = s.ss_customer_sk 
       AND s.ss_sold_date_sk = d.d_date_sk 
       AND s.ss_item_sk = i.i_item_sk 
       AND d.d_month_seq = (SELECT DISTINCT ( d_month_seq ) 
           FROM   date_dim 
           WHERE  d_year = 1998 
                  AND d_moy = 7) 
       AND i.i_current_price > 1.2 * 
            (SELECT Avg(j.i_current_price) 
                     FROM   item j 
                     WHERE  j.i_category = i.i_category) 
GROUP  BY a.ca_state 
HAVING Count(*) >= 10 
ORDER  BY cnt
LIMIT 100; 
\end{lstlisting}
\end{minipage}

\noindent
\begin{lstlisting}[language=sql, breaklines=true, caption={Complex analytical query over the LHBench TPC-DS dataset.}, label={lst:lhbench-complex-analytical}, captionpos=b]
--Query 8 (TPC-DS Query 4)
WITH year_total 
     AS (SELECT c_customer_id      customer_id, 
                c_first_name       customer_first_name, 
                c_last_name        customer_last_name, 
                c_preferred_cust_flag
                    customer_preferred_cust_flag, 
                c_birth_country
                    customer_birth_country, 
                c_login            customer_login, 
                c_email_address
                    customer_email_address, 
                d_year             dyear, 
                Sum(( ( ss_ext_list_price
                    - ss_ext_wholesale_cost
                    - ss_ext_discount_amt) 
                    + ss_ext_sales_price ) / 2) year_total,
                's'                sale_type 
         FROM   customer, 
                sales_store store_sales, 
                date_dim 
         WHERE  c_customer_sk = ss_customer_sk 
                AND ss_sold_date_sk = d_date_sk 
         GROUP  BY c_customer_id, 
                   c_first_name, 
                   c_last_name, 
                   c_preferred_cust_flag, 
                   c_birth_country, 
                   c_login, 
                   c_email_address, 
                   d_year 
         UNION ALL 
         SELECT c_customer_id            customer_id, 
                c_first_name             customer_first_name, 
                c_last_name              customer_last_name, 
                c_preferred_cust_flag
                    customer_preferred_cust_flag, 
                c_birth_country         customer_birth_country, 
                c_login                  customer_login, 
                c_email_address
                    customer_email_address, 
                d_year                   dyear, 
                Sum(( ( (cs_ext_list_price 
                     - cs_ext_wholesale_cost 
                     - cs_ext_discount_amt) 
                     + cs_ext_sales_price) / 2 )) year_total,
                'c'                      sale_type 
         FROM   customer, 
                catalog_sales, 
                date_dim 
         WHERE  c_customer_sk = cs_bill_customer_sk 
                AND cs_sold_date_sk = d_date_sk 
         GROUP  BY c_customer_id, 
                   c_first_name, 
                   c_last_name, 
                   c_preferred_cust_flag, 
                   c_birth_country, 
                   c_login, 
                   c_email_address, 
                   d_year 
         UNION ALL 
         SELECT c_customer_id            customer_id, 
                c_first_name             customer_first_name, 
                c_last_name              customer_last_name, 
                c_preferred_cust_flag
                    customer_preferred_cust_flag, 
                c_birth_country
                    customer_birth_country, 
                c_login      customer_login, 
                c_email_address
                    customer_email_address, 
                d_year                   dyear, 
                Sum(( ( ( ws_ext_list_price 
                 - ws_ext_wholesale_cost 
                 - ws_ext_discount_amt) 
                 + ws_ext_sales_price ) / 2 )) year_total, 
                'w'                      sale_type 
         FROM   customer, 
                web_sales, 
                date_dim 
         WHERE  c_customer_sk = ws_bill_customer_sk 
                AND ws_sold_date_sk = d_date_sk 
         GROUP  BY c_customer_id, 
                   c_first_name, 
                   c_last_name, 
                   c_preferred_cust_flag, 
                   c_birth_country, 
                   c_login, 
                   c_email_address, 
                   d_year) 
SELECT t_s_secyear.customer_id, 
               t_s_secyear.customer_first_name, 
               t_s_secyear.customer_last_name, 
               t_s_secyear.customer_preferred_cust_flag 
FROM   year_total t_s_firstyear, 
       year_total t_s_secyear, 
       year_total t_c_firstyear, 
       year_total t_c_secyear, 
       year_total t_w_firstyear, 
       year_total t_w_secyear 
WHERE  t_s_secyear.customer_id = t_s_firstyear.customer_id 
       AND t_s_firstyear.customer_id 
            = t_c_secyear.customer_id 
       AND t_s_firstyear.customer_id 
            = t_c_firstyear.customer_id 
       AND t_s_firstyear.customer_id 
            = t_w_firstyear.customer_id 
       AND t_s_firstyear.customer_id 
            = t_w_secyear.customer_id 
       AND t_s_firstyear.sale_type = 's' 
       AND t_c_firstyear.sale_type = 'c' 
       AND t_w_firstyear.sale_type = 'w' 
       AND t_s_secyear.sale_type = 's' 
       AND t_c_secyear.sale_type = 'c' 
       AND t_w_secyear.sale_type = 'w' 
       AND t_s_firstyear.dyear = 2001 
       AND t_s_secyear.dyear = 2001 + 1 
       AND t_c_firstyear.dyear = 2001 
       AND t_c_secyear.dyear = 2001 + 1 
       AND t_w_firstyear.dyear = 2001 
       AND t_w_secyear.dyear = 2001 + 1 
       AND t_s_firstyear.year_total > 0 
       AND t_c_firstyear.year_total > 0 
       AND t_w_firstyear.year_total > 0 
       AND CASE 
             WHEN t_c_firstyear.year_total > 0 
                THEN t_c_secyear.year_total / 
                  t_c_firstyear.year_total 
             ELSE NULL 
           END > CASE 
                   WHEN t_s_firstyear.year_total > 0 THEN 
                   t_s_secyear.year_total / 
                   t_s_firstyear.year_total 
                   ELSE NULL 
                 END 
       AND CASE 
             WHEN t_c_firstyear.year_total > 0 
                THEN t_c_secyear.year_total / 
                  t_c_firstyear.year_total 
             ELSE NULL 
           END > CASE 
                   WHEN t_w_firstyear.year_total > 0 THEN 
                   t_w_secyear.year_total / 
                   t_w_firstyear.year_total 
                   ELSE NULL 
                 END 
ORDER  BY t_s_secyear.customer_id, 
          t_s_secyear.customer_first_name, 
          t_s_secyear.customer_last_name, 
          t_s_secyear.customer_preferred_cust_flag
LIMIT 100; 
\end{lstlisting}

\noindent
\begin{minipage}{\linewidth}
\begin{lstlisting}[language=sql, breaklines=true, caption={Analytical queries over the NYC Yellow Taxi dataset.}, label={lst:nyctaxi-analytical}, captionpos=b]
--Based on https://learn.microsoft.com/en-us/sql/machine-learning/tutorials/demo-data-nyctaxi-in-sql.
--Query 9
SELECT DISTINCT passengerCount
    , ROUND (SUM (fareAmount),0) as TotalFares
    , ROUND (AVG (fareAmount),0) as AvgFares
FROM dropoff_pickup
GROUP BY passengerCount
ORDER BY  AvgFares DESC;

--Query 10
SELECT * FROM dropoff_pickup LIMIT 10;

--Query 11
SELECT COUNT(*) FROM dropoff_pickup;
\end{lstlisting}
\end{minipage}

\end{document}